\documentclass[11pt]{article}
\usepackage[utf8]{inputenc}
\usepackage{amssymb,amsmath,amsthm,mathtools,hyperref}
\usepackage{color}
\usepackage[justification=RaggedRight]{caption}
\usepackage{lipsum}
\usepackage{graphicx}
\graphicspath{ {images/} }
\usepackage{tikz}
\usepackage{pgfplots}
\usepackage{algorithm,algorithmic}
\usepackage{listings}
\usepackage{authblk}
\pgfplotsset{compat=1.8}
\newtheorem{lemma}{Lemma}

\newtheorem{definition}{Definition}
\newtheorem{proposition}{Proposition}
\newtheorem{remark}{Remark}
\newtheorem{notation}{Notation}
\newcommand\numberthis{\addtocounter{equation}{1}\tag{\theequation}}

\title{Bayesian Alternatives to the Black-Litterman Model}
\author[1]{Mihnea S. Andrei}
\author[1]{John S.J. Hsu}
\affil[1]{Statistics and Applied Probability,
  University of California, Santa Barbara}
\affil{{mandrei@ucsb.edu},hsu@pstat.ucsb.edu}

\date{\today}

\begin{document}

\maketitle

\begin{abstract}
The Black-Litterman model combines investors' personal views with historical data and gives optimal portfolio weights. In this paper we will introduce the original Black-Litterman model (section \ref{original BL}), we will modify the model such that it fits in a Bayesian framework by considering the investors' personal views to be a direct prior on the means of the returns and by adding a typical Inverse Wishart prior on the covariance matrix of the returns (section \ref{inv wishart P inv}). Lastly, we will use Leonard and Hsu's (1992)\cite{hsu1992} idea of adding a prior on the logarithm of the covariance matrix (section \ref{log(sigma) prior}). Sensitivity simulations for the level of confidence that the investor has in their own personal views were performed and performance of the models was assessed on a test data set consisting of returns over the month of January $2018$.
\end{abstract}

\section{The Original Black-Litterman Model}\label{original BL}
The Black-Litterman model was developed in the early $1990$'s and has been widely used for asset allocation. This model attempts to combine the market equilibrium \footnote{Please refer to subsection \ref{pi} for more details on how the market equilibrium is computed} with the investor's personal views. It can be shown that the optimal portfolio for an investor is the sum of a portfolio proportional to the market equilibrium portfolio and a weighted sum of portfolios reflecting their views. Please see section \ref{personal views} for an example of how personal views are created and section \ref{introduction to model} for more details on the model.

\subsection{Estimating the Market Equilibrium} \label{pi}

The market is in equilibrium when all investors hold the market portfolio, $w_{eq}$. It is when the demand for the assets in this portfolio equals the supply. If we denote by $\pi$ the market equilibrium returns, then the CAPM equation is $\pi=\lambda \Sigma w_{eq}$. Here, $\lambda$ is the investor's risk aversion coefficient and $\Sigma$ is the covariance matrix of the returns on the assets in the portfolio\cite{litterman2002}.

\subsection{Example of Personal Views} \label{personal views}
Let us see how personal views are inputted in the traditional model. For example, let us consider 4 assets: AAPL, AMZN, GOOG and MSFT. Also, let us consider that we believe that AAPL will outperform AMZN by $2\%$  and we also believe that GOOG will have returns that amount to $5\%$. The columns in $P$ represent the 4 stocks in the order in which we enumerated them previously. Each row in $P$ and $q$ represents a personal view:

\[
P=
\begin{bmatrix}
 1 & -1 & 0 & 0 \\
 0 & 0  & 1 & 0
\end{bmatrix},
q=
\begin{bmatrix}
 0.02 \\
 0.05
\end{bmatrix}
\]

Each personal view has associated with it an uncertainty that the investor has with respect to the view. The measures of confidence are entered as diagonal entries in a matrix $\Omega$. As we will see in the next section when we present the assumptions of the model (equations (\ref{TraditionalBL}), $\Omega$ is a covariance matrix. Hence, on the main diagonal we will have the variances of the returns for our personal views. Therefore, a small value reflects a high confidence in the view and vice versa (one approach of obtaining $\Omega$ is by simply manually imputing them and another approach is to estimate it). 

\subsection{The Black-Litterman Approach}\label{introduction to model}

Now that we have seen what are the individual pieces of the model, we are ready to present the mathematical formulation. We will consider that the investor is looking at $n$ assets and has $k$ different views on those assets.

The return of the assets is considered to be random, $r \sim N(\mu,\Sigma)$. 

We create a prior on the mean of this return: $\mu \sim N(\pi,\tau \Sigma)$. However, the proof would work in the exact same way for any general covariance matrix (please see Appendix \ref{proof traditional} for the proof). The variable $\pi$ represents the market equilibrium returns and it is obtained by using an equation equivalent to the CAPM (\ref{eq capm}): $\pi=\lambda\Sigma w_{eq}$, with $\lambda$ been the investor's risk aversion parameter, $w_{eq}$ been the market equilibrium weights and $\Sigma$ the covariance matrix. Also, here $\tau$ is considered to be a parameter that reflects the uncertainty in the CAPM prior. Notice that the smaller the $\tau$, the closer our $\mu$ will be to the market equilibrium returns $\pi$.

Besides this prior, we also have the investor's personal views: $P\mu \sim N(q,\Omega)$, where the same notation as in the previous section \ref{personal views} is used \cite{litterman2002}.

Hence, the model is represented by the following $3$ distributions, where the last 2 are priors: 
\begin{equation}	
\label{TraditionalBL}
\begin{aligned} 
r \sim N(\mu,\Sigma) \\
\mu \sim N(\pi,\tau \Sigma) \\		
P\mu \sim N(q,\Omega)
\end{aligned} 
\end{equation}

By combining the $2$ priors (the last $2$ distributions from above) and using that as a new prior on the distribution of the return, one can show that the posterior of the return is (please see the appendix \ref{proof traditional} for the proof): 
\begin{gather*}
r\sim N(\bar{\mu},M^{-1}+\Sigma)
\end{gather*}
Where the notation used is, 
\begin{gather}
M^{-1}=\left((\tau\Sigma)^{-1}+ P^T\Omega^{-1}P \right)^{-1} \\
\bar{\mu}=\left((\tau\Sigma)^{-1}+ P^T\Omega^{-1}P \right)^{-1}\left( (\tau\Sigma)^{-1}\pi+P^T\Omega^{-1}q \right)
\end{gather}
Also, let $\bar{\Sigma}=M^{-1}+\Sigma$

Coming back to the example of personal views given in the previous section, if the investor thinks that GOOG will outperform the other 3 companies, it is enough to put a long weight (positive weight) of $1$ on GOOG and short weights (negative weights) for the other $3$. In the Black-Litterman model, the return is considered to be random and we have just seen that the posterior distribution is also normal, but with a mean represented by equation $(3)$. This equation appropriately takes into account market volatility and correlations also. Let us further look at the weights, $w$, that one would obtain when using the posterior of the returns. The typical approach to the problem that an investor with risk aversion parameter $\lambda$ has when trying to maximize the returns of the portfolio while minimizing the risk is to maximize the function $w^T\bar{\mu}-\frac{\lambda}{2}w^T\bar{\Sigma}w$ with respect to $w$. By taking the derivative with respect to $w$, we obtain for the optimal weights an equation equivalent to CAPM: 
\begin{gather} \label{eq capm}
\overline{\mu}=\lambda \overline{\Sigma} w \Leftrightarrow  w=\frac{1}{\lambda}\bar{\Sigma}^{-1}\bar{\mu} 
\end{gather}
Using equation $(3)$ and the identity $(\Sigma+M^{-1})^{-1}=M-M(M+\Sigma^{-1})^{-1}M$, one can show that:
\begin{gather*}
w^*=\frac{1}{1+\tau}(w_{eq}+P^T\times\Delta) \\
\Delta=\tau\Omega^{-1}\frac{q}{\lambda}-A^{-1}P\frac{\Sigma}{1+\tau}w_{eq}-A^{-1}P\frac{\Sigma}{1+\tau}P^T\tau\Omega^{-1}\frac{q}{\lambda} \\
A=\frac{\Omega}{\tau}+P\frac{\Sigma}{1+\tau}P^T
\end{gather*}

The first equation from above shows that the Black-Litterman optimal portfolio weights are the equilibrium weights plus a weighted sum of the personal views ($P^T$ has as columns the views). Moreover, the three terms in the second equation have interpretations also: 
\begin{itemize}
\item The first term shows that the more importance is given to a view either by having a high return $q_k$ (the $k^{th}$ entry of $q$), or by having a smaller uncertainty $\omega_k$ (this is the $k^{th}$ diagonal entry in $\Omega$ and hence it will be equal to $\frac{1}{\omega_k}$ in $\Omega^{-1}$), the more weight it will carry into the optimal portfolio
\item The second term shows that a weight is penalized if the value of the product between the covariance and the market equilibrium return is high. This would indicate that the view carries less additional information.
\item The last term shows that an optimized weight is penalized if the covariance between different portfolio views is high. This makes sense since it would mean that different views add little new information.
\end{itemize}

One can also observe the fact that if an investor has a different risk aversion parameter, $\hat{\lambda}$, he can obtain the optimized portfolio weights by using the equation $\hat{w^*}=\frac{\lambda}{\hat{\lambda}}w^*$

\section{Inverse Wishart prior on $\Sigma$} \label{inv wishart P inv}

\subsection{Introduction}

First of all, we have to observe what we would like to potentially improve on the original approach. But what are the shortcomings of the original approach? 
\begin{itemize}
\item	From the proof presented in Appendix \ref{proof traditional}, we notice that it is not a true Bayesian approach since, as mentioned in the previous section, the distributions represented by the last $2$ equations in (\ref{TraditionalBL}) are combined to find a new prior for $\mu$. This is used as a prior for the return. 
\item	The original approach does not consider any current returns and we would like to model returns over a certain period of time (which will be the investor's investment horizon). Hence, we would like to have $r_1,r_2,...,r_m \sim N_n(\mu, \Sigma)$
\item	The covariance matrix $\Sigma$ is estimated from historical data in the original approach. Hence, we would like to put prior on the distribution of $\Sigma$ and the most common approach for covariance matrices when the data is from a normal distribution is to use an Inverse Wishart conjugate prior. 
\end{itemize}

Let $r_1,r_2,...,r_m \sim N_n(\mu,\Sigma)$ be the observed returns over a period of length $m$ for the $n$ assets that the investor is considering. Also, we would like to consider a prior on the mean of those returns that is projected directly by our personal views $P$. Therefore, we would have, similarly to the traditional model, $P\mu \sim N(q,\Omega)$. The difference is that we would like to add an additional Inverse Wishart prior $\Sigma \sim W^{-1}(\nu,\Sigma_0)$. Since the returns $r\in \mathbb{R}^{n}$, while the prior is on $P\mu \in \mathbb{R}^{k}$ ($k$ is the number of views), the prior is not fully specified, which would suggest the idea that we would need to create an invertible $P$. Also, once we have an invertible $P$, we can follow two approaches:
\begin{itemize}
\item Obtain the distribution of $\mu$, which could be easily done if $P$ is invertible.
\item From the very beginning transform the returns into the personal view space: $r_i^*=Pr_i$. This procedure will still require $P$ to be invertible since after obtaining the posterior in the transformed space, we have to be able to transform back. 
\end{itemize}
Hence, either way, we would need to have a matrix $P$ that is invertible and this brings us to the following discussion.

\subsection{Creating an Invertible $P$} \label{P inv}

The matrix of our personal views is very likely not to be invertible since, as we have seen in Section \ref{personal views}, we can have relative views (the ones for which the rows sum up to $0$) and we can have absolute views (only one $1$ in a row). Moreover, the $k$ views that we will have (the number of rows in $P$) will be smaller than the $n$ assets that we are considering to trade (the number of columns in $P$). In this section, we will present a method in which we can add rows to $P$ such that the resulting square matrix $P^*$ is invertible. The main idea is based on the way in which one would row reduce a matrix to the echelon form. 

It is well known that a matrix is invertible iff its row reduced echelon form is the identity matrix. This gives us the idea of taking our matrix $P$ and adding rows to it in order to make it invertible: 
\begin{itemize}
\item For each column in $P$ that has only $0$'s, we have to create a new row that will have only one $1$ in the respective column and $0$'s in all the others.
\item If a row has more than $1$ nonzero entry, for each one except the entries in the pivot columns, we have to create a row in which we have a $1$.
\end{itemize}

For example, if we consider the matrix $P$ from section \ref{personal views}, the above procedure gives us: 

\begin{equation*}
P=
\begin{bmatrix}
 1 & -1 & 0 & \color{red}0 \\
 0 & 0  & 1 & \color{red}0
\end{bmatrix}\rightarrow
\begin{bmatrix}
 1 & \color{red}-1 & 0 & 0 \\
 0 & 0  & 1 & 0 \\
 0 & 0 & 0 & 1 
\end{bmatrix}\rightarrow
\begin{bmatrix}
 1 & -1 & 0 & 0 \\
 0 & 0  & 1 & 0 \\
 0 & 0 & 0 & 1 \\
 0 & 1 & 0 & 0 
\end{bmatrix}=P^*=
\begin{bmatrix}
P \\
P_2
\end{bmatrix}
\end{equation*}
Please notice that we denoted by $P^*$ the augmented invertible matrix based on $P$, and by $P_2$ the part that was added to $P$.

We have seen at the beginning of the section that the model equations are: 
\begin{gather*}
r_1,r_2,...,r_m|\mu,\Sigma \sim N_n(\mu,\Sigma) \\
P\mu \sim N_k(q,\Omega) \\
\Sigma \sim W^{-1}(\nu,\Sigma_0)
\end{gather*}

As mentioned previously, we have $2$ approaches. We choose to transform the returns using our personal views and, therefore, let $r_i^*=P^*r_i$. Hence, $r_1^*,r_2^*,...,r_m^* \sim N_n(\mu^*,\Sigma^*)$, where $\mu^*=P^*\mu$ and $\Sigma^*=P^*\Sigma{P^*}^T$. Now, after the transformation onto the space of personal views, the 3 equations become: 

\begin{gather*}
r_1^*,r_2^*,...,r_m^* \sim N_n(\mu^*,\Sigma^*) \\
\mu^* \sim N_n(q^*,\Omega^*) \\
\Sigma^* \sim W^{-1}(\nu,\Sigma_0)
\end{gather*}

However, just like $\mu^*, \Sigma^*$ were clearly depending on $\mu, \Sigma$, we also have that $q^*, \Omega^*$ depend on the originals $q, \Omega$:

\begin{gather} \label{q*}
q^*=E \left[\mu^* \right]=E\left[ P^*\mu \right]=
E \left[\begin{bmatrix}
P \\
P_2
\end{bmatrix}\mu\right]=
\begin{bmatrix}
q \\
q_2
\end{bmatrix}
\end{gather}
\begin{gather} \label{Omega*}
\Omega^*=Var \left( P^*\mu \right)=Var \left( \begin{bmatrix}
P \\
P_2
\end{bmatrix}\mu \right)= 
\begin{bmatrix}
\Omega & P Var(\mu) P_2^T \\
P_2 Var(\mu) P^T & P_2 Var(\mu) P_2^T
\end{bmatrix}
\end{gather}

\subsection{Derivation of Posterior Distributions}

Now that we have found a method to augment $P$ to a matrix $P^*$ that is invertible and we also managed to create corresponding $q^*$ and $\Omega^*$, the problem is posed in a more typical Bayesian framework: 

\begin{equation}\label{ass1}
\begin{aligned}
r_1^*,r_2^*,...,r_m^* | \mu^*,\Sigma^* \sim N_n(\mu^*,\Sigma^*)
\end{aligned}
\end{equation}
\begin{equation}\label{ass2}
\begin{aligned}
\mu^* \sim N(q^*,\Omega^*)
\end{aligned}
\end{equation}
\begin{equation}\label{ass3}
\begin{aligned}
\Sigma^* \sim W^{-1}(\nu,\Sigma_0)
\end{aligned}
\end{equation}

From (\ref{ass1}), we obtain that the joint density of our returns is: 
\begin{align*}
 f(r_1^*,...,r_m^*|\mu^*,\Sigma^*)\propto det(\Sigma^*)^{-\frac{m}{2}}exp\left \{ -\frac{1}{2}\sum_{i=1}^m (r_i^*-\mu^*)^T{\Sigma^*}^{-1}(r_i^*-\mu^*)\right\}
\end{align*}

From (\ref{ass2}), we obtain that the density for $\mu^*$ is: 
\begin{align*}
 f(\mu^*) \propto det(\Omega^*)^{-\frac{1}{2}}exp \left \{ -\frac{1}{2} (\mu^*-q^*)^T{\Omega^*}^{-1}(\mu^*-q^*) \right \}
\end{align*}

Similarly, using (\ref{ass3}), we obtain that the density for $\Sigma^*$ is: 
\begin{align*}
 f(\Sigma^*) \propto det(\Sigma^*)^{-\frac{\nu+n+1}{2}}exp\left \{ -\frac{1}{2} Tr\left(\Sigma_0{\Sigma^*}^{-1}\right)\right \}
\end{align*}

Hence, by multiplying the above 3 equations, we obtain that the joint density for all of them is: 
\begin{equation} \label{joint dist}
\begin{aligned}
 f(r_1^*,...,r_m^*,\mu^*,\Sigma^*)\propto det(\Sigma^*)^{-\frac{\nu+m+n+1}{2}}exp \left \{ -\frac{1}{2} Tr\left(\Sigma_0{\Sigma^*}^{-1}\right)\right \} det(\Omega^*)^{-\frac{1}{2}}\cdot \\
exp\left \{-\frac{1}{2}\left((\mu^*-q^*)^T{\Omega^*}^{-1}(\mu^*-q^*)+ \sum_{i=1}^m (r_i^*-\mu^*)^T{\Sigma^*}^{-1}(r_i^*-\mu^*) \right ) \right \}
\end{aligned}
\end{equation}

Let us focus on the parenthesis in the second exponential and let us prove the following result. 

\begin{lemma} \label{r-mu identity}
The following equality holds, where $\bar r^*=\frac{\sum_{i=1}^m r_i^*}{m}$:
\begin{align*}
 \sum_{i=1}^m (r_i^*-\mu^*)^T{\Sigma^*}^{-1}(r_i^*-\mu^*)=&\sum_{i=1}^m (r_i^*-\bar r^*)^T{\Sigma^*}^{-1}(r_i^*-\bar r^*)+ \\
 &+m(\bar r^*-\mu^*)^T{\Sigma^*}^{-1}(\bar r^*-\mu^*)
\end{align*}
\end{lemma}
\begin{proof}
 We will start by manipulating the right hand side: 
 \begin{align*}
  RHS&=\sum_{i=1}^m({r_i^*}^T{\Sigma^*}^{-1}r_i^*-{r_i^*}^T{\Sigma^*}^{-1}\bar r^*-{\overline{r}^{*}}^{T}{\Sigma^*}^{-1}{r_i^*}+{\overline{r}^*}^T{\Sigma^*}^{-1}\overline{r}^*)+\\
  &m{\overline{r}^*}^T{\Sigma^*}^{-1}\bar r^*
     -m{\overline{r}^*}^T{\Sigma^*}^{-1}\mu^*-m{\mu^*}^T{\Sigma^*}^{-1}\bar r^*+m{\mu^*}^T{\Sigma^*}^{-1}\mu^*
  \end{align*}
   
   But since $m\bar r^*=\sum_{i=1}^m r_i^* \Rightarrow m{\overline{r} ^*}^T=\sum_{i=1}^m {r_i^*}^T$, we obtain that:
  
  \begin{align*}
  RHS&=\sum_{i=1}^m({r_i^*}^T{\Sigma^*}^{-1}r_i^*-{r_i^*}^T{\Sigma^*}^{-1}\bar r^*-{\overline{r}^*}^T{\Sigma^*}^{-1}r_i^*)+2m{\overline{r}^*}^T{\Sigma^*}^{-1}\bar r^*-\\
     &-\left (\sum_{i=1}^m {r_i^*}^T\right){\Sigma^*}^{-1}\mu^*-{\mu^*}^T{\Sigma^*}^{-1}\left(\sum_{i=1}^m r_i^* \right)+\sum_{i=1}^m{\mu^*}^T{\Sigma^*}^{-1}\mu^*=\\
     &=\sum_{i=1}^m({r_i^*}^T{\Sigma^*}^{-1}r_i^*-{r_i^*}^T{\Sigma^*}^{-1}\bar r^*-{\overline{r}^*}^T{\Sigma^*}^{-1}r_i^*)+2m{\overline{r}^*}^T{\Sigma^*}^{-1}\bar r^*-\\
     &-\sum_{i=1}^m {r_i^*}^T{\Sigma^*}^{-1}\mu^*-\sum_{i=1}^m{\mu^*}^T{\Sigma^*}^{-1} r_i^* +\sum_{i=1}^m{\mu^*}^T{\Sigma^*}^{-1}\mu^*
  \end{align*}
  
  We observe that ${\Sigma^*}^{-1}$ and $\bar r^*$ do not depend on the sum. Hence, we can factor them out:
  
  \begin{align*}
  RHS&=\sum_{i=1}^m({r_i^*}^T{\Sigma^*}^{-1}r_i^*-{r_i^*}^T{\Sigma^*}^{-1}\mu^*-{\mu^*}^T{\Sigma^*}^{-1}r_i^*+ {\mu^*}^T{\Sigma^*}^{-1}\mu^*)+\\
     &+2m{\overline{r}^*}^T{\Sigma^*}^{-1}\bar r^*-\left( \sum_{i=1}^m {r_i^*}^T\right){\Sigma^*}^{-1}\bar r^*- \bar r^*{\Sigma^*}^{-1}\left(\sum_{i=1}^m r_i^* \right)=\\
     &=\sum_{i=1}^m (r_i^*-\mu^*)^T{\Sigma^*}^{-1}(r_i^*-\mu^*)+2m{\overline{r}^*}^T{\Sigma^*}^{-1}\bar r^*- m{\overline{r}^*}^T{\Sigma^*}^{-1}\bar r^*- \\
     &-m{\overline{r}^*}^T{\Sigma^*}^{-1}\bar r^*=\sum_{i=1}^m (r_i^*-\mu^*)^T{\Sigma^*}^{-1}(r_i^*-\mu^*)
 \end{align*}

\end{proof}

Let us make a notation before we proceed: $s^2=\frac{\sum_{i=1}^m (r_i^*-\bar r^*)^T{\Sigma^*}^{-1}(r_i^*-\bar r^*)}{m-1}$ 

Now, by using \textbf{Lemma 1}, we are ready to come back to the paranthesis in the second exponential from the joint density of $(r_1^*,...,r_m^*,\mu^*,\Sigma^*)$ (equation (\ref{joint dist})):
\begin{gather*}
 (\mu^*-q^*)^T{\Omega^*}^{-1}(\mu^*-q^*)+ \sum_{i=1}^m (r_i^*-\mu^*)^T{\Sigma^*}^{-1}(r_i^*-\mu^*)=\\
 =(m-1)s^2+ m(\bar r^*-\mu^*)^T{\Sigma^*}^{-1}(\bar r^*-\mu^*)+(\mu^*-q^*)^T{\Omega^*}^{-1}(\mu^*-q^*)=\\
 =(m-1)s^2+(\bar r^*-\mu^*)^T(m{\Sigma^*}^{-1})(\bar r^*-\mu^*)+ (\mu^*-q^*)^T{\Omega^*}^{-1}(\mu^*-q^*) \numberthis \label{eq lemma 1}
\end{gather*}

\begin{lemma}\textbf{(Completing the square)} \label{lemma2}
For any $A \in {\mathbb{R}}^{p\times p}$ positive definite, $B \in {\mathbb{R}}^{p\times p}$ positive semi-definite and $a,b\in \mathbb{R}^{p}$ the following identity holds: 
\begin{gather*} 
 (y-a)^TA(y-a)+(y-b)^TB(y-b)=(y-y^*)^T(A+B)(y-y^*)+\\
 +(a-b)^TH(a-b),
\end{gather*}
where $y^*=(A+B)^{-1}(Aa+Bb)$ and $H=A(A+B)^{-1}B$. If, furthermore, $B$ is positive definite, then $H=(A^{-1}+B^{-1})^{-1}$.    \cite{hsu1999} 
\end{lemma}

Since both of our normal distributions are not degenerated because we can have inverses for both $\Sigma^*$ and $\Omega^*$, we conclude that they do not have any eigenvalues equal to $0$. Moreover, since they are covariance matrices, we know that they are positive semi-definite. Therefore their eigenvalues are greater than or equal to $0$. But since they can't be $0$, we observe that they have to be strictly greater than $0$. This implies that both matrices are positive definite and therefore we can use the second formula for $H$ in \textbf{Lemma 2}.

Now we are ready to apply this result to equation (\ref{eq lemma 1}) for $y=\mu^*$, $a=\bar r^*$, $b=q^*$, $A=m{\Sigma^*}^{-1}$ and $B={\Omega^*}^{-1}$:
\begin{gather*}
 (\ref{joint dist})\Leftrightarrow (\ref{eq lemma 1}) \Leftrightarrow (m-1)s^2+(\mu^*-\overline{\mu^*})^T(m{\Sigma^*}^{-1}+{\Omega^*}^{-1})(\mu^*-\overline{\mu^*})+ \\
 +(\bar r^*-q^*)^TH(\bar r^*-q^*),
\end{gather*}
where $\overline{\mu^*}=(m{\Sigma^*}^{-1}+{\Omega^*}^{-1})^{-1}(m{\Sigma^*}^{-1}\bar r^*+{\Omega^*}^{-1} q^*)$ and $H=\left(\frac{1}{m}\Sigma^*+\Omega^*\right)^{-1}$.

If we go back with this result in the joint density represented by equation (\ref{joint dist}), we obtain that: 
\begin{gather*}
 f(r_1^*,...,r_m^*,\mu^*,\Sigma^*)\propto \\ 
 \propto det(\Sigma^*)^{-\frac{m}{2}}exp\left \{ -\frac{1}{2}(\mu^*-\overline{\mu^*})^T(m{\Sigma^*}^{-1}+{\Omega^*}^{-1})(\mu^*-\overline{\mu^*})\right \}\cdot \\ 
 \cdot exp\left \{-\frac{1}{2}(\bar r^*-q^*)^TH(\bar r^*-q^*)+(m-1)s^2 \right \} \cdot\\ 
 \cdot det(\Omega^*)^{-\frac{1}{2}}det(\Sigma^*)^{\frac{\nu+n+1}{2}}exp\left \{-\frac{1}{2}Tr\left(\Sigma_0{\Sigma^*}^{-1}\right) \right \}
\end{gather*}

Since the only part that depends on $\mu^*$ is the first line of the above equation, we conclude that: 

\begin{gather*} 
 f(\mu^*|r_1^*,...,r_m^*,\Sigma^*)\propto exp\left \{ -\frac{1}{2}(\mu^*-\overline{\mu^*})^T (m{\Sigma^*}^{-1}+{\Omega^*}^{-1})(\mu^*-\overline{\mu^*})\right \}\Rightarrow \\
 \mu^*|r_1^*,...,r_m^*,\Sigma^* \sim N_n \left(\overline{\mu^*},\overline{\Sigma^*}\right), \text{ where } \\
 \overline{{\mu^*}}=\left(m{\Sigma^*}^{-1}+{\Omega^*}^{-1}\right)^{-1}\left(m{\Sigma^*}^{-1}\bar r^*+{\Omega^*}^{-1} q^*\right)\\
 \overline{\Sigma^*}=\left(m{\Sigma^*}^{-1}+{\Omega^*}^{-1}\right)^{-1} \numberthis \label{posterior mu}
\end{gather*}

In order to find the posterior of $\Sigma^*$, it is easier to start from the original joint density represented by equation (\ref{joint dist}). By collecting the terms that depend on $\Sigma^*$ we obtain:
\begin{gather*}
 f(\Sigma^*|r_1^*,...,r_m^*,\mu^*)\propto det(\Sigma^*)^{-\frac{\nu+m+n+1}{2}}\\ 
 exp \left \{ -\frac{1}{2} \left ( \sum_{i=1}^m (r_i^*-\mu^*)^T{\Sigma^*}^{-1}(r_i^*-\mu^*)+ Tr\left(\Sigma_0{\Sigma^*}^{-1}\right) \right) \right \} \numberthis \label{eq8}
\end{gather*}
We notice that this is quite close to another Inverse Wishart distribution, the only step left that we have to make is to manipulate the exponential. Let us notice that 
\begin{gather*}
 \sum_{i=1}^m (r_i^*-\mu^*)^T{\Sigma^*}^{-1}(r_i^*-\mu^*)\in \mathbb{R} \Rightarrow \sum_{i=1}^m (r_i^*-\mu^*)^T{\Sigma^*}^{-1}(r_i^*-\mu^*)= \\ 
 =Tr\left (\sum_{i=1}^m (r_i^*-\mu^*)^T{\Sigma^*}^{-1}(r_i^*-\mu^*) \right)
 =\sum_{i=1}^m Tr \left((r_i^*-\mu^*)^T{\Sigma^*}^{-1}(r_i^*-\mu^*)\right)
\end{gather*}
But inside the $Trace$ matrices are cyclically commutative as long as the dimensions agree: 
\begin{gather*}
 \sum_{i=1}^m Tr \left((r_i^*-\mu^*)^T{\Sigma^*}^{-1}(r_i^*-\mu^*)\right)
 =\sum_{i=1}^m Tr \left((r_i^*-\mu^*)(r_i^*-\mu^*)^T{\Sigma^*}^{-1}\right)= \\
 =Tr \left(\sum_{i=1}^m(r_i^*-\mu^*)(r_i^*-\mu^*)^T{\Sigma^*}^{-1}\right)
\end{gather*}
Finally, by using this result and equation (\ref{eq8}), we obtain:
\begin{gather*}
 (\ref{eq8})\Leftrightarrow f(\Sigma^*|r_1^*,...,r_m^*,\mu^*)\propto det(\Sigma^*)^{-\frac{\nu+m+n+1}{2}}\\
 exp \left \{ -\frac{1}{2} Trace\left ( \left (\Sigma_0+\sum_{i=1}^m (r_i^*-\mu^*)(r_i^*-\mu^*)^T \right ){\Sigma^*}^{-1}\right) \right \}
\end{gather*}

We notice that this is the kernel of an Inverse Wishart distribution. Therefore, we can conclude that: 
\begin{gather} \label{sigma posterior}
 \Sigma^*|r_1^*,...,r_m^*,\mu^* \sim W^{-1}\left(\nu+m,\Sigma_0+\sum_{i=1}^m (r_i^*-\mu^*)(r_i^*-\mu^*)^T\right)
\end{gather}
 
Now that we have the posterior distributions, we can implement a Gibbs Sampler, which we will see in the following section, where we will also look at how the parameters of the model were estimated. 

\subsection{Implementation}   \label{implementation}

For implementation purposes, $4$ stocks were chosen: Apple(AAPL), Amazon(AMZN), Google(GOOG) and Microsoft(MSFT). Closing prices for the $4$ from 1/2/2015 until 5/1/2017 were considered and the returns were computed. Now, this data is split into $2$ parts, one representing the \textit{current data} (the last $m$ returns $r_1,...,r_m$, here $m=21$) and the rest representing \textit{historical data} used to estimate the parameters in the model. The reason why $m=21$ was chosen is because we are thinking of modeling the returns that happen within a period of approximately a month and $21$ is the average number of trading days in a month. Hence, in this example, the trading period for such an investor would be over a month. Next step is to augment $P$ as discussed in section \ref{P inv}. Once $P^*$ is created, we can just create our transformed returns $r_i^*=P^*r_i$. For this example, the personal views were (the columns represent AAPL, AMZN, GOOG, MSFT, respectively): 

\[
P=
\begin{bmatrix}
 1 & -1 & 0 & 0 \\
 0 & 0  & 1 & -1
\end{bmatrix},
q=
\begin{bmatrix}
 0.02 \\
 0.05
\end{bmatrix}
\]

If we look at the second assumption in the model represented by equation (\ref{ass2}), we notice that $q^*$ and $\Omega^*$ are, respectively, the mean and covariance matrix for $\mu^*$, which is in turn a mean of returns from a particular month (again, in this example $m=21$, approximately a month). Hence, one solution for estimating the parameters would be to take the returns from each month in the historical data and to compute their means. This way, we would have estimates for the monthly mean returns $\hat{\mu^*_i}$, with $i$ an integer between $1$ and the number of months in the historical data. Once we obtain those, we can estimate $\hat{q^*}$ and $\hat{\Omega^*}$ by taking the mean and the covariance of $\hat{\mu^*_i}$.

But, we have to remember that we need to reflect our personal views in the estimation presented above. In equations (\ref{q*}) and (\ref{Omega*}), we have showed how one should combine the estimates from the procedure just presented with the investor's personal views: 
\begin{itemize}
\item Equation (\ref{q*}) shows that we should take the $\hat{q^*}$ obtained through the above estimation and replace the first $k$ entries with $q$ ($k$, as mentioned at the beginning, was the number of personal views).
\item Equation (\ref{Omega*}) shows that we should take the obtained $\hat{\Omega^*}$ and replace the top left $k \times k$ matrix with our personal choice of $\Omega$.
\end{itemize}

Now that the parameters of our model are estimated, a typical Gibbs Sampler was used based on the posteriors represented by equations (\ref{sigma posterior}) and (\ref{posterior mu}).

\begin{algorithm}
\caption{Gibbs Sampler}
\begin{algorithmic}[1]
\scriptsize 
\STATE $ {\Sigma^*}^{(t+1)}|r_1^*,...,r_m^*,{\mu^*}^{(t)} \sim W^{-1}\left(\nu+m,\Sigma_0+\sum_{i=1}^m (r_i^*-{\mu^*}^{(t)})(r_i^*-{\mu^*}^{(t)})^T\right)$
\STATE  \begin{gather*}
 {\mu^*}^{(t+1)}|r_1^*,...,r_m^*,{\Sigma^*}^{(t+1)} \sim N_n \left({\overline{\mu^*}}^{(t+1)},{\overline{\Sigma^*}}^{(t+1)}\right), \text{ where } \\
 \overline{{\mu^*}^{(t+1)}}=\left(m{{\Sigma^*}^{(t+1)}}^{-1}+{\Omega^*}^{-1}\right)^{-1}\left(m{{\Sigma^*}^{(t+1)}}^{-1}\bar r^*+{\Omega^*}^{-1} q^*\right)\\
 {\overline{\Sigma^*}^{(t+1)}}=\left(m{{\Sigma^*}^{(t+1)}}^{-1}+{\Omega^*}^{-1}\right)^{-1} 
\end{gather*}
\end{algorithmic}
\end{algorithm}

A burning period of $10^3$ was chosen and the number of iterations for the Gibbs Sampler is $10^4$. After the Gibbs sampler is completed, one would only have to take the mean of the simulated ${\mu^*}^{(t)}$, call it $\hat{\overline{\mu^*}}$, and the average of the simulated ${\Sigma^*}^{(t)}$, call it $\hat{\overline{\Sigma^*}}$. However, one has to remember that those were transformed using $P^*$, hence now we would have to transform them back into the original space: $\hat{\overline{\mu}}={P^*}^{-1}\hat{\overline{\mu^*}}, \hat{\overline{\Sigma}}={P^*}^{-1}\hat{\overline{\Sigma^*}}{P^*}^{-T}$. Just like in the original model, in order to get the weights, one would use an equation similar to the CAPM one presented in section \ref{introduction to model}: $w=\frac{1}{\lambda}\hat{\overline{\Sigma}}^{-1}\hat{\overline{\mu}}$. Here, $\lambda=2.5$, as chosen in the original model. Also there has been extensive research when it comes to choosing $\lambda$. For trading stocks a risk aversion coefficient between $2$ and $3$ is reasonable.\cite{janecek2004}
Finally, we are ready to compare the results obtained under the original model with the ones obtained from this one. 

\subsection{Results Comparison} \label{results comparison}

Before we delve into how we compare the $2$ approaches, let us make the observation that in order to make any kind of comparison, one has to make sure that the same data sets were used and the parameters were estimated in the same way. Albeit the same personal views were imputed (same $P$, $\Omega$, $q$), the two approaches differ in the fact that the alternative one has a prior on $\Sigma$ and the original one makes use of the market equilibrium returns, which is estimated using $\pi=\lambda\Sigma w_{eq}$. In the following table, we ca look at the setup for both side by side: 
\begin{table}[H]
\centering
\begin{tabular}{lr}
Alternative & Original \\
\hline
$\begin{aligned}
r_1^*,r_2^*,...,r_m^* \sim N_n(\mu^*,\Sigma^*) \\
\mu^* \sim N(q^*,\Omega^*)
\end{aligned}$ & 
$\begin{aligned}
 r \sim N(\mu,\Sigma) \\
\mu \sim N(\pi,\tau \Sigma)
\end{aligned}$
\end{tabular}
\end{table}
Instead of the market equilibrium, the alternative approach simply has another parameter, which is estimated as mentioned in section \ref{implementation} (also the alternative has a prior on $\Sigma$ and takes into consideration current data). Besides this difference, the two are using the same data sets and the same parameters. 
Now, the question becomes how should one compare the two. One obvious approach would be to see how the two would perform if one would use them on the real market, which will be presented in the results section for the models that will follow later in this paper. However, it is of more interest to us to check how close to our personal opinion is the posterior mean obtained from the Gibbs Sampler. 
\begin{remark} \label{remark 1}
Since for both models we have that $P\mu \sim N(q,\Omega)$, the smaller the uncertainty in our views (the diagonal entries of $\Omega$), the smaller the standard deviation and, hence, the more certain the investor is about that particular view. 
\end{remark}
Hence, from the above remark, we will look at how $P \hat{\bar{\mu}}$ behaves as we look at small values for the diagonal entries of $\Omega$. But how should one define "small"? As we have seen in section \ref{implementation}, the expected returns for the views were 
$q=
\begin{bmatrix}
 0.02 \\
 0.05
\end{bmatrix}$.
Hence, even a value of $10^{-4}$ is quite large since this would be the variance of our view and, therefore, the standard deviation would become $10^{-2}$. Hence, a $95\%$ confidence interval for the first view would be $(0,0.04)$. If one tries to input even smaller $\omega$, the Inverse Wishart random generator gives a non-singularity error. Hence, we conclude that we compare the models on values of the diagonal of the matrix $\Omega$ that are between $0$ and $10^{-4}$. Albeit we can't input smaller $\omega$, for the purposes of checking the following remark, we changed $q$ to $q=\begin{bmatrix}
0.2 \\
0.5
\end{bmatrix}$. Hence, for both models an exhaustive method was implemented that would compute for each pair of diagonal entries in $\Omega$ a posterior mean $\hat{\overline{\mu}}$. Once this is obtained, the distance $|P\hat{\overline{\mu}}-q|$ can be calculated for both models. 
\begin{remark} \label{remark 2}
Since $P\mu \sim N(q,\Omega)$, we have that $\lim_{\Omega \to O_2}P\mu=q$ a.s.
\end{remark}
Therefore, as the diagonal entries of $\Omega$ get smaller and smaller we expect to get closer and closer to $q$. 
\begin{remark}
If we look at the posterior of $\mu^*$ we have that: 
\begin{gather*}
 \mu^*|r_1^*,...,r_m^*,\Sigma^* \sim N_n \left(\overline{\mu^*},\overline{\Sigma^*}\right), \text{ where } \\
 \overline{{\mu^*}}=\left(m{\Sigma^*}^{-1}+{\Omega^*}^{-1}\right)^{-1}\left(m{\Sigma^*}^{-1}\bar r^*+{\Omega^*}^{-1} q^*\right)\\
 \overline{\Sigma^*}=\left(m{\Sigma^*}^{-1}+{\Omega^*}^{-1}\right)^{-1}
\end{gather*}
If we consider a small $\Omega^*\Rightarrow {\Omega^*}^{-1}$ is large and therefore the whole term $m{\Sigma^*}^{-1}+{\Omega^*}^{-1} \approx {\Omega^*}^{-1}\Rightarrow (m{\Sigma^*}^{-1}+{\Omega^*}^{-1})^{-1} \approx \Omega^*$. Similarly, $(m{\Sigma^*}^{-1}\bar r^*+{\Omega^*}^{-1} q^*) \approx {\Omega^*}^{-1} q^*$ for small enough $\Omega^*$. Hence, we would expect that the mean of the simulated ${\mu^*}^{(t)}$ is close to $q^*$. Or, with the notation already used, $\hat{\overline{\mu^*}}\approx q^*$. Hence, by using the previous remark also, we obtain that $P({P^*}^{-1}\overline{\mu^*})\approx q$.
\end{remark}

The following graphs have as $2$ of the axis the $2$ diagonal entries in $\Omega$ and the third one represents the distance $|P\overline{\mu}-q|=|P({P^*}^{-1}\overline{\mu^*})-q|$ :

\begin{figure}[ht]
        \begin{minipage}[b]{0.45\linewidth}
            \centering
            \includegraphics[width=\textwidth]
{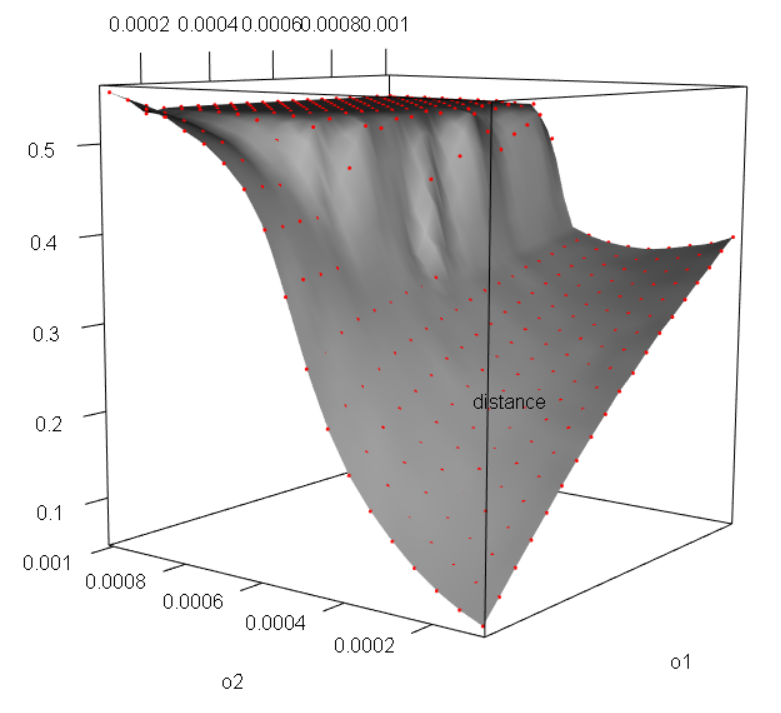}
            \caption{Results of $\Omega$ for the alternative model}
        \end{minipage}
        \hspace{0.5cm}
        \begin{minipage}[b]{0.45\linewidth}
            \centering
            \includegraphics[width=\textwidth]{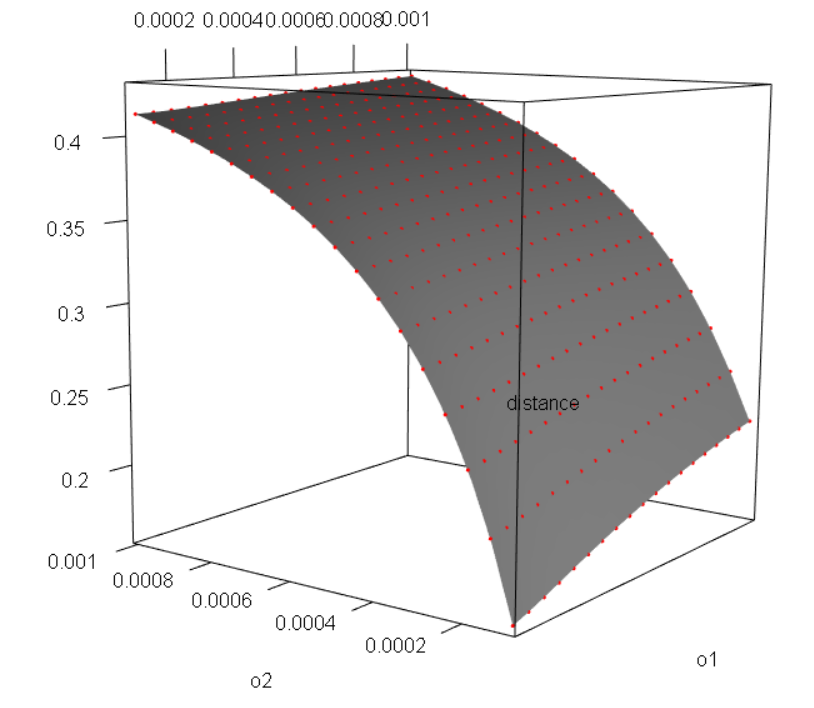}
            \caption{Results of $\Omega$ for original model}
        \end{minipage}
\end{figure}

We notice from the z-axis, which represents the distance mentioned above, that the modified model more closely follows the personal views.  We can look at some specific values of the distance for different pairs of $\omega_1$ and $\omega_2$ in table \ref{table 1}. 

\begin{table}[h!]
\centering
 \begin{tabular}{||c c c c||} 
 \hline
 $\omega_1$ & $\omega_2$ & Original & Alternative \\ [0.5ex] 
 \hline\hline
 $10^{-4}$ & 0.0001 & 0.154 & 0.052 \\ 
 \hline
 $10^{-4}$ & 0.00015 & 0.2 & 0.063 \\
 \hline
 $10^{-4}$ & 0.0002 & 0.235 & 0.077 \\
 \hline
 $10^{-4}$ & 0.00025 & 0.263 & 0.118 \\
 \hline
 $10^{-4}$ & 0.0003 & 0.286 & 0.118 \\  
 \hline
 $10^{-4}$ & 0.00035 & 0.305 & 0.145 \\ 
 \hline
 $10^{-4}$ & 0.0004 & 0.321 & 0.177 \\
 \hline
 $10^{-4}$ & 0.00045 & 0.335 & 0.222 \\
 \hline
 $10^{-4}$ & 0.0005 & 0.347 & 0.282 \\
 \hline
 $10^{-4}$ & 0.00055 & 0.357 & 0.354 \\ [1ex] 
 \hline
\end{tabular}
\caption{Table with specific distance values}
\label{table 1}
\end{table}

However, we would like to see if the structure of $P\hat{\overline{\mu}}$ is similar to $q$. For this we keep the two entries in $\Omega$ equal, we exhaustively search over small $\omega$ s.t. $\Omega=\omega\mathbb{I}$ and we plot the 2 entries of $P\hat{\overline{\mu}}$ together with the respective $\omega$. Please note that the blue point in figures \ref{alt vs orig} represent the exact value of $q=
\begin{bmatrix}
 0.2 \\
 0.5
\end{bmatrix}$, which would be obtained for $\omega=0$.

By comparing the $2$ figures, we notice that not only the point simulations represented by the red points are closer, but the whole curve (which was obtained by interpolation) seems to be closer to the theoretical value represented by the blue point. Also, we notice that in both cases, as $\omega$ increases, $P\hat{\overline{\mu}}$ gets further away from $q$, which is what theoretically should happen. 

\begin{figure}[H]
        \begin{minipage}[b]{0.45\linewidth}
            \centering
            \includegraphics[width=\textwidth]{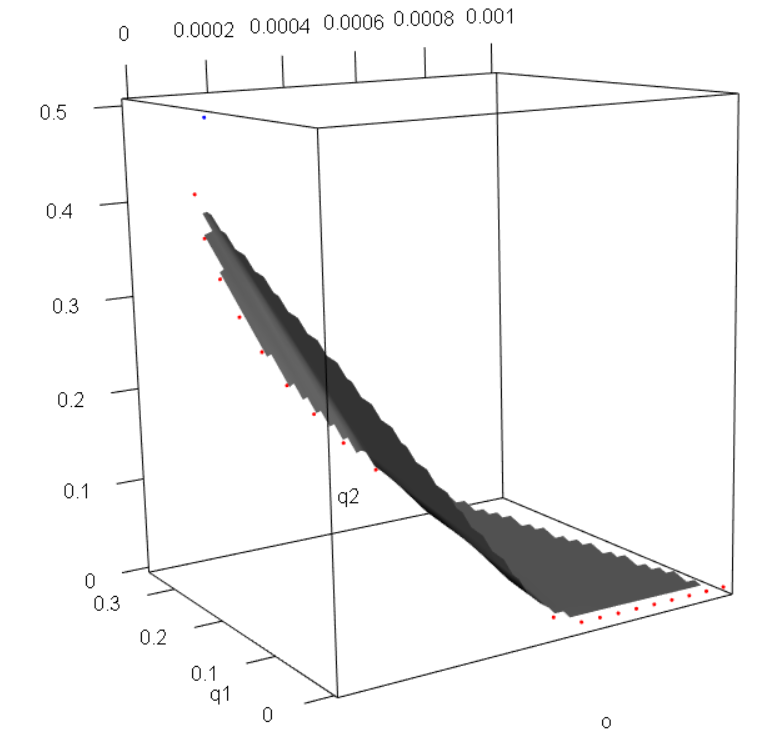}
            \caption{Results of $\Omega$ for the alternative model}
        \end{minipage}
        \hspace{0.5cm}
        \begin{minipage}[b]{0.45\linewidth}
            \centering
            \includegraphics[width=\textwidth]{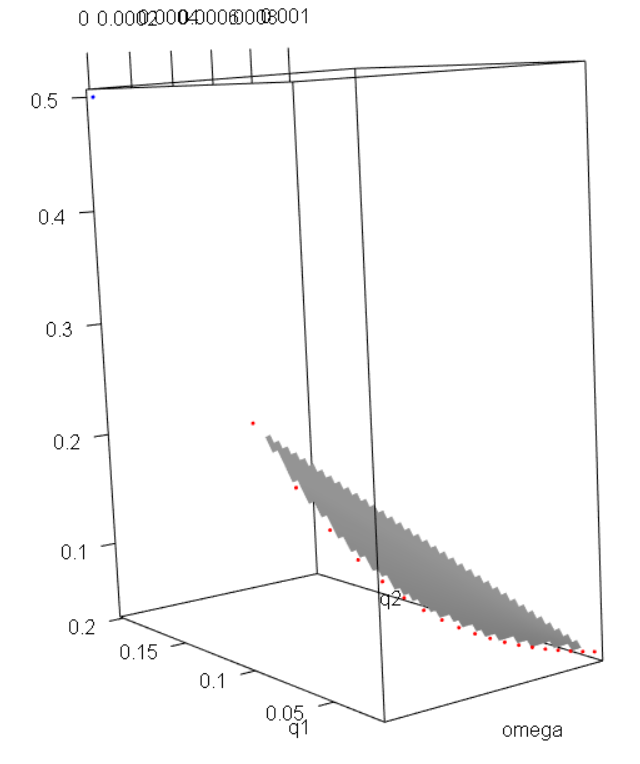}
            \caption{Results of $\Omega$ for original model}
        \end{minipage}
\end{figure} \label{alt vs orig}

\subsection{But do we need an invertible $P$?}

In section \ref{P inv}, we introduced a method of creating an invertible matrix $P$ by adding rows. Ideally, the matrix of views should not be modified in any way. Also, in general, one can't create an invertible matrix by adding rows, especially when the number of views is close to the number of stocks selected. Also, the method of augmenting $P$ is not unique. Therefore, we would like to be able to derive the posteriors when $P$ is unchanged from what the investor is inputting. Hence, in this section, we will consider the same setup as before, the only difference being the fact that $P$ is not even square: 

\begin{gather*}
r_1,r_2,...,r_m|\mu,\Sigma \sim N_n(\mu,\Sigma) \\
P\mu \sim N_k(q,\Omega) \\
\Sigma \sim W^{-1}(\nu,\Sigma_0)
\end{gather*}

Since $P$ shows up in the second equation of our model assumptions, the only posterior that will change from what we had previously will be that for $\mu$. Hence, in the joint distribution, we will consider only the terms depending on $\mu$: 

\begin{gather*}
f(\mu|r_1,...,r_m)\propto exp \left\lbrace  -\frac{1}{2}  \sum_{i=1}^m (r_i-\mu)^T \Sigma^{-1} (r_i-\mu) \right\rbrace \cdot \\
\cdot exp \left\lbrace -\frac{1}{2}(P\mu-q)^T \Omega^{-1} (P\mu-q)  \right\rbrace
\end{gather*}

For the first exponential we can use Lemma \ref{r-mu identity}. This yields: 
\begin{gather*}
f(\mu|r_1,...,r_m)\propto exp \left\lbrace  -\frac{1}{2} \left( (m-1)s^2+m(\overline{r}-\mu)^T\Sigma^{-1}(\overline{r}-\mu) \right)  \right\rbrace \cdot \\
\cdot exp \left\lbrace -\frac{1}{2}(q-P\mu)^T \Omega^{-1} (q-P\mu)  \right\rbrace
\end{gather*}

We remember that $s^2=\frac{\sum_{i=1}^m (r_i-\bar r)^T{\Sigma}^{-1}(r_i-\bar r)}{m-1}$ and hence this term does not depend on $\mu$. Now, let us focus on the remaining terms in the exponential: 
\begin{gather*}
(\overline{r}-\mu)^T(m\Sigma^{-1})(\overline{r}-\mu)+ (q-P\mu)^T \Omega^{-1} (q-P\mu)=\\
=\overline{r}^T(m\Sigma^{-1})\overline{r}-2\overline{r}^T(m\Sigma^{-1})\mu+\mu^T(m\Sigma^{-1})\mu+q^T\Omega^{-1}q-2q^T\Omega^{-1}P\mu+\\
+\mu^TP^T\Omega^{-1}P\mu=\mu^T\left( m\Sigma^{-1}+P^T\Omega^{-1}P \right)\mu-2\left( \overline{r}^T(m\Sigma^{-1})+q^T\Omega^{-1}P \right)\mu+\\
+\overline{r}^T(m\Sigma^{-1})\overline{r}+q^T\Omega^{-1}q
\end{gather*}

Since only the first two terms depend on $\mu$, we obtain that: 
\begin{gather*}
f(\mu|r_1,...,r_m,\Sigma)\propto exp \left\lbrace  -\frac{1}{2} \mu^T(m\Sigma^{-1}+P^T\Omega^{-1}P)\mu \right\rbrace \cdot \\
\cdot exp \left\lbrace -\frac{1}{2}2(m\Sigma^{-1}\overline{r}+P^T\Omega^{-1}q)^T\mu \right\rbrace
\end{gather*}

\begin{lemma}
Let $M$ be a symmetric and invertible matrix, then the following identity holds: 
\begin{gather*}
x^TMx-2b^Tx=(x-M^{-1}b)^TM(x-M^{-1}b)-b^TM^{-1}b
\end{gather*}
\end{lemma}
\begin{proof}
We just need to expand the quadratic term:
\begin{gather*}
(x-M^{-1}b)^TM(x-M^{-1}b)=x^TMx-2b^TM^{-1}Mx+b^TM^{-1}MM^{-1}b=\\
=x^TMx-2b^Tx+b^TM^{-1}b
\end{gather*}
\end{proof}

Hence, if we apply this lemma for $x=\mu$, $M=m\Sigma^{-1}+P^T\Omega^{-1}P$ and $b=m\Sigma^{-1}\overline{r}+P^T\Omega^{-1}q$, we obtain that the exponential in the distribution of the posterior of $\mu$ is (the $-\frac{1}{2}$ still sits in front of the formula, we just omit it in the following for simplicity of writing): 
\begin{gather*}
\left( \mu-(m\Sigma^{-1}+P^T\Omega^{-1}P)^{-1}(m\Sigma^{-1}\overline{r}+P^T\Omega^{-1}q) \right)^T(m\Sigma^{-1}+P^T\Omega^{-1}P)\cdot \\
\cdot \left( \mu-(m\Sigma^{-1}+P^T\Omega^{-1}P)^{-1}(m\Sigma^{-1}\overline{r}+P^T\Omega^{-1}q) \right)-b^TM^{-1}b
\end{gather*}

Lastly, we notice that $b$ and $M$ do not depend on $\mu$, and hence, the posterior of $\mu$ is dictated by the first big term, which is actually the density of a normal distribution: 
\begin{gather*}
\boxed{
\mu|r_1,...,r_m,\Sigma \sim N\left(\mu_{post}, \Sigma_{post} \right)} ,\text{ where } \\
\boxed{
\mu_{post}=(m\Sigma^{-1}+P^T\Omega^{-1}P)^{-1}(m\Sigma^{-1}\overline{r}+P^T\Omega^{-1}q)} \\
\boxed{
\Sigma_{post}=\left( m\Sigma^{-1}+P^T\Omega^{-1}P \right)^{-1}
}
\end{gather*}

This posterior is very close to the one obtained by using the first approach (represented by equation (\ref{posterior mu})), the only difference being the fact that in this new approach the matrix $P$ shows up. This is because here we did not change the investor inputted matrix $P$, while in the previous approach we augmented $P$ in order for it to be invertible. 

\subsection{Implementation}

Implementing this model is straightforward since it is very similar to the previous version. The only difference is the fact that in the posterior for $\mu$ we have $P$ appearing, while in the previous model there was no $P$ since we were adding rows to it so that it becomes invertible. We remind ourselves that this was the first approach because we can take the inverse and easily find the prior distribution of $\mu$ from the prior distribution of $P\mu$. Using the derived posteriors, the Gibbs Sampler is:

\begin{algorithm}
\caption{Gibbs Sampler}
\begin{algorithmic}[1]
\scriptsize 
\STATE $ {\Sigma}^{(t+1)}|r_1,...,r_m,{\mu}^{(t)} \sim W^{-1}\left(\nu+m,\Sigma_0+\sum_{i=1}^m (r_i-{\mu}^{(t)})(r_i-{\mu}^{(t)})^T\right)$
\STATE  \begin{gather*}
{\mu}^{(t+1)}|r_1,...,r_m,{\Sigma}^{(t+1)} \sim N\left({\mu_{post}}^{(t+1)}, {\Sigma_{post}}^{(t+1)} \right) ,\\
{\mu_{post}}^{(t+1)}=(m{\Sigma^{(t+1)}}^{-1}+P^T\Omega^{-1}P)^{-1}(m{\Sigma^{(t+1)}}^{-1}\overline{r}+P^T\Omega^{-1}q) \\
{\Sigma_{post}}^{(t+1)}=\left( m{\Sigma^{(t+1)}}^{-1}+P^T\Omega^{-1}P \right)^{-1}
\end{gather*}
\end{algorithmic}
\end{algorithm}

\subsection{Results} \label{results P nonsqr Inv Wish}

Just like before, we will try to look at the sensitivity of our model to different confidence levels. Remarks \ref{remark 1} and \ref{remark 2} made when we presented the results for the previous model still hold. Since in $\Omega$ we have on the main diagonal (call them $\omega_i$) the variances in our views $P\mu$, the smaller the $\omega_i$, the more certain we are in view $i$. This should also be reflected in our posterior: if we provide very large $\omega_i$, it means that we are very uncertain about the views and the model should take into consideration the history a lot more, while if we provide very small values for $\omega_i$, it means that we are very certain about the views and the model should take them into consideration a lot more than the history. 

Just like before, in order to quantify and visualize the model's sensitivity to different confidence levels, we will look at the distance $|P\mu_{post}-q|$ (which will be on one of the axis in our plots) over different combinations of $\omega_i$. The same $4$ stocks from before were chosen (AAPL,AMZN,GOOG,MSFT), but since this work is more recent, the daily returns are from 1/2/2014 to 12/29/2017. The views are (rows are views and the columns represent the 4 stocks in the order AAPL,AMZN,GOOG,MSFT): 

\begin{gather*}
q=\begin{bmatrix}
0.02\\
0.05
\end{bmatrix},
P=\begin{bmatrix}
-1 & 1 & 0 & 0\\
0 & 0 & 1 & -1\\
\end{bmatrix}
\end{gather*}

When it comes to the confidence levels in the $2$ views, one can input values as small as $10^{-7}$ without encountering any numerical issues, like we did previously when we were augmenting the matrix $P$. Hence, one doesn't need to make any change to the model when implementing it or when inputting any value. We take a grid of equally spaced points $(\omega_1,\omega_2)$ between $10^{-7}$ and $2\cdot 10^{-5}$. The burn period was set to $10^3$ and the number of iterations in the Gibbs Sampler was set to $10^4$.

However, one could also use the same views, but considering the daily returns for the whole $S\&P 500$ instead of just for $4$ stocks. For this, we need the daily returns of companies actively traded in $S\&P 500$ over the period mentioned before. We won't have to change $q$ at all, but $P$ has more columns since they would represent the stocks in the famous index and it will still have $2$ rows for the same $2$ views. One would fill out $P$ by making sure that in the first row and the column corresponding to AAPL we will have a $-1$, in the first row and the column corresponding to AMZN we will have a $1$ and similarly for the second row. Of course, the dimension of some of the matrices and vectors will be much bigger and therefore, all computations will be more expensive. Hence, this version was parallelized and the number of iterations in the Gibbs Sampler decreased to $10^3$ (as we will see, even with so few iterations, convergence for the mean is achieved, but convergence for the covariance matrix is not). The interval $10^{-7}$ to $10^{-5}$ for the confidence levels was split into $4$ parts, in the following way: 

\begin{gather*}
(\omega_1,\omega_2)\in \begin{cases}
\left( 10^{-6}\rightarrow 10^{-5},10^{-6}\rightarrow 10^{-5} \right)\\
\left( 10^{-5}\rightarrow 10^{-4},10^{-5}\rightarrow 10^{-4} \right)\\
\left( 10^{-6}\rightarrow 10^{-5},10^{-5}\rightarrow 10^{-4} \right)\\
\left( 10^{-5}\rightarrow 10^{-4},10^{-6}\rightarrow 10^{-5} \right)
\end{cases}
\end{gather*}

Each one of the $4$ ranges from above was divided into a grid of $7^2$ points. Each point was ran on one core, taking a little more than $4$ hours to run. 

\begin{itemize}
\item	When $\omega_1=10^{-6}$, a $95\%$ confidence interval for the first view would be $(0.018,0.022)$, which would show that the investor is very confident.
\item	When $\omega_1=10^{-4}$, a $95\%$ confidence interval for the first view would be $(0,0.04)$, which would show that the investor is not as confident.
\end{itemize}

\begin{figure}[ht]
        \begin{minipage}[b]{0.45\linewidth}
            \centering
            \includegraphics[width=\textwidth]
{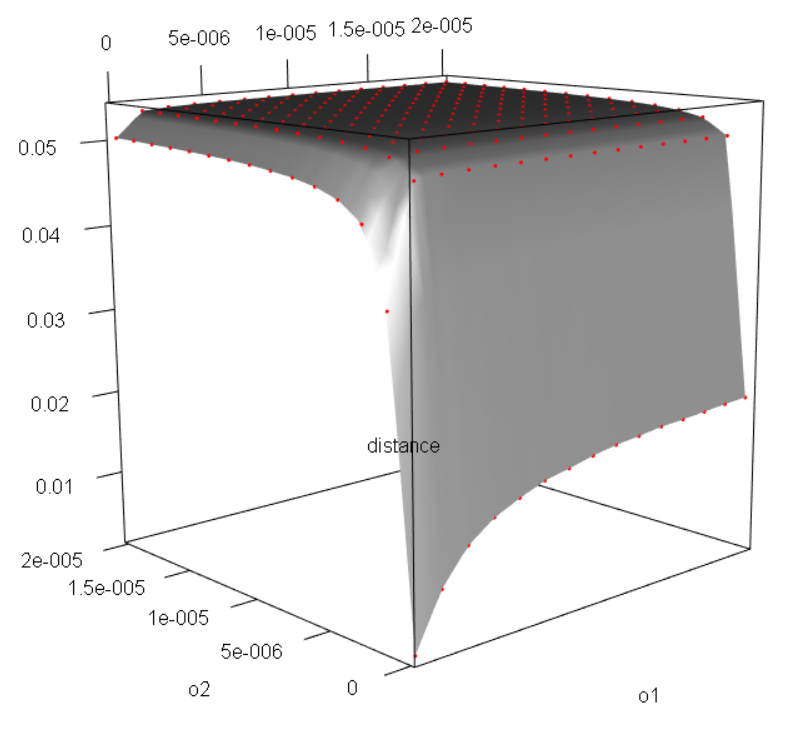}
            \caption{$|P\mu_{post}-q|$ when taking only the $4$ stocks}
        \end{minipage}
        \hspace{0.5cm}
        \begin{minipage}[b]{0.45\linewidth}
            \centering
            \includegraphics[width=\textwidth]{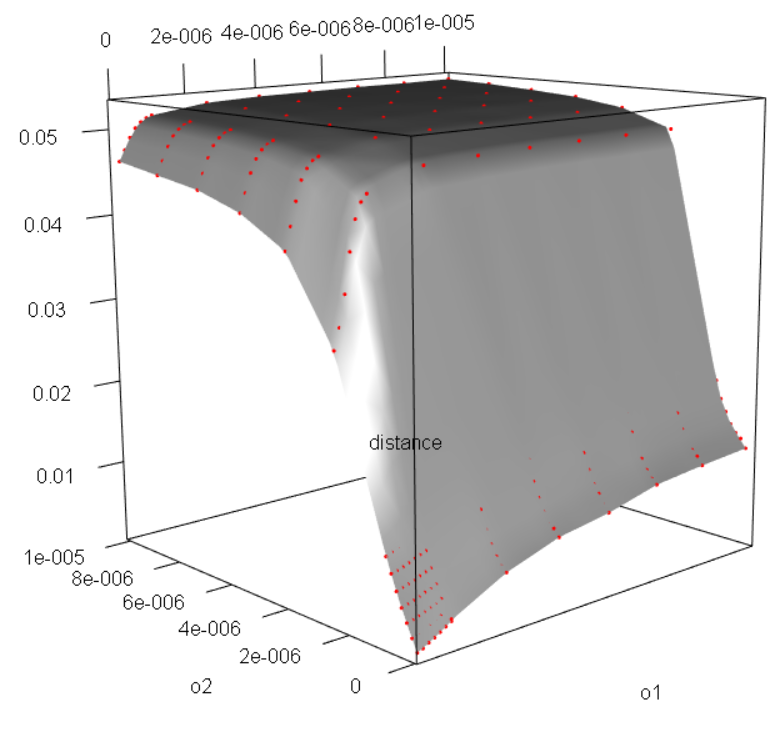}
            \caption{$|P\mu_{post}-q|$ when taking $S\&P 500$}
        \end{minipage}
\end{figure}

In the figures presented, we notice that both curves have similar shapes, albeit the one on the right converges faster to $0$ as $o_i$ become smaller ($\omega_i$ in our model). Also, the curve on the right seems to be underneath the one on the left. Intuitively, this is because there is a lot more information in our prior for $\Sigma$ when we take the whole $S\&P 500$. 
Moreover, both have very similar shapes. The distances go to $0$ as $\omega_i$ go to $0$. This is in tune with our intuition of how the model should behave like: as one gets more and more confident in their inputted views, the model should put a lot of importance on them and not on the historical data. Vice-versa, in both figures the distance seems to converge to a certain value as $\omega_i$ become bigger and bigger. Again, this is what we would think that the model should do since large $\omega_i$, suggests that one is uncertain about the personal view and therefore, the history should play a more important role. Indeed, if we would only take the historical returns, an unbiased estimate for $\mu$ is $\overline{r}$ and the distance becomes $|P\overline{r}-q|=0.05388875$, which is what the plots seem to tend to converge to.

We will move our focus towards looking at the profits (or losses) that one would obtain when using the model to trade over the month of January 2018 (testing data consisting of daily returns between 1/2/2018 and 1/30/2018) using an initial capital of $100,000\$$ (this does not include any capital requirements for short selling). We remember that in order to get portfolio weights we use the same approach as before. From Gibbs Sampling we estimate $\mu_{post}$ and $\Sigma_{post}$ and we use the CAPM equation \ref{eq capm}: $w=\frac{1}{2.5}{\Sigma_{post}}^{-1}\mu_{post}$.

Albeit when we took the whole $S\&P 500$ the number of iterations in the Gibbs Sampler was small, we notice from the above analysis that we still get very good estimates for $\mu_{post}$ since the posterior distance behaves exactly like our intuition suggests it should do. The running averages for the mean also converge fast for small $\omega_i$. However, because of the size of $\Sigma_{post}$ and because of the fact that one has to take its inverse in order to compute the portfolio weights $w$, the number of iterations is not enough to give accurate predictions of profits. Nevertheless, for completeness, the average profit when considering the whole $S\&P 500$ is $13191.39\$$ with a standard deviation of $2908.134\$$. 

We will now present the profits obtained when using only $4$ stocks. We notice that the first view has a bigger impact on the profits curve than the second view. Moreover, as the confidence in the first view increases (as $\omega_1$ goes to $0$), the profits sky rocket. This is because over the month of January 2018 AMZN outperformed AAPL by $23.997\%$ and our view was indeed that AMZN will overrun AAPL (albeit by only $2\%$, a $10^{th}$ of what actually happened in reality).

\begin{figure}[ht]
	\centering
        \begin{minipage}[b]{0.6\linewidth}
            \centering
            \includegraphics[width=\textwidth]
{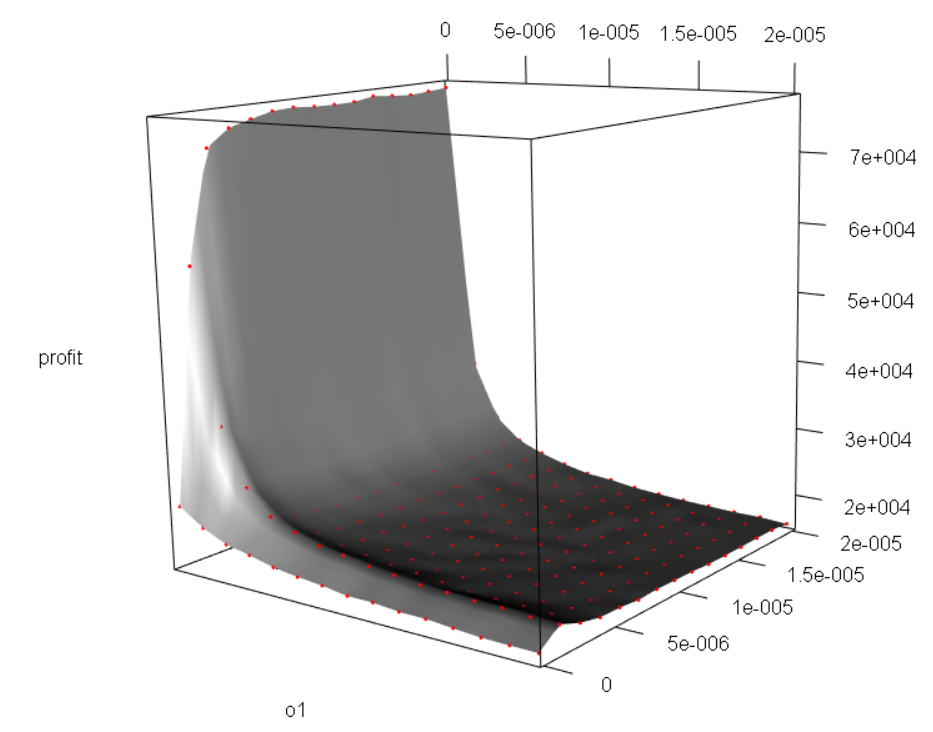}
            \caption{Profit when taking only $4$ stocks}
        \end{minipage}
\end{figure}

AMZN outperforming AAPL by nearly $24\%$ in one month is uncommon. Therefore, next we will present the same results, the only change made is that we replace AMZN with FB (Facebook). The same data sets were used and all other inputs stay exactly the same as we just presented at the beginning of this section, except $q$. We will also look at how the model behaves when the investor inputs a personal view exactly like what happened during the month of January 2018 (very "informed" investor) and exactly the opposite of what happened during the month of January (very "uninformed" investor). Therefore, we will also look at what happens when we choose $q=\begin{bmatrix}
0.06212815\\
0.01366718
\end{bmatrix}$ and $q=-\begin{bmatrix}
0.06212815\\
0.01366718
\end{bmatrix}$, respectively. 

\begin{figure}[ht]
        \begin{minipage}[b]{0.29\linewidth}
            \centering
            \includegraphics[width=\textwidth]
{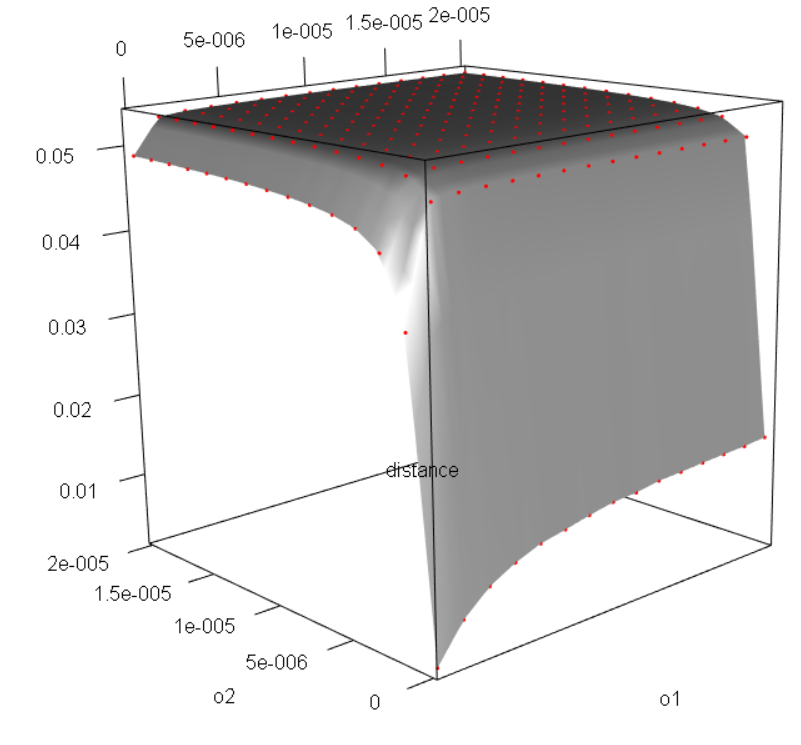}
            \caption{$4$ stocks,FB in and $q={[0.02,0.05]}^T$}
        \end{minipage}
        \hspace{0.5cm}
        \begin{minipage}[b]{0.29\linewidth}
            \centering
            \includegraphics[width=\textwidth]{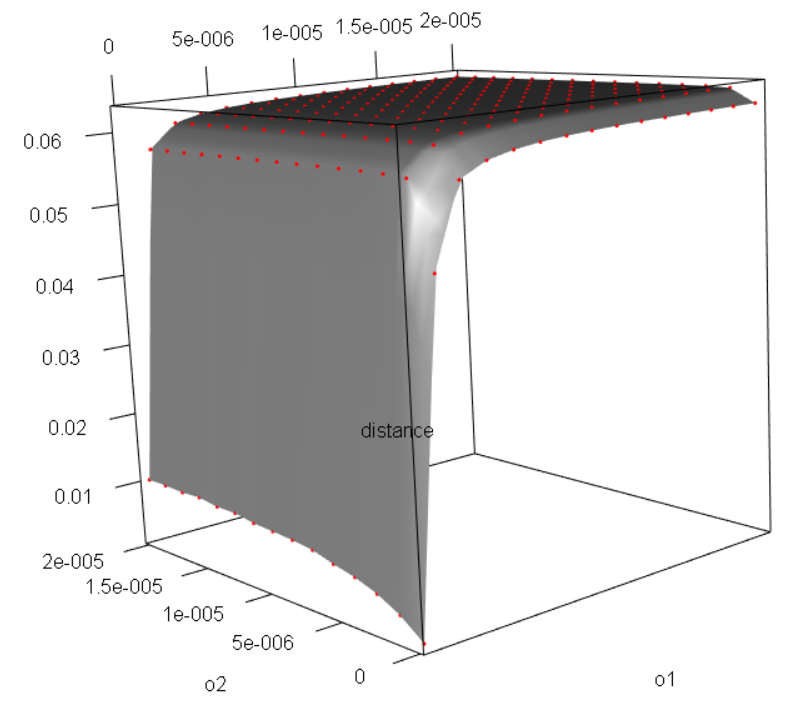}
            \caption{$4$ stocks, FB in and view exactly like reality}
        \end{minipage}
        \hspace{0.5cm}
        \begin{minipage}[b]{0.29\linewidth}
            \centering
            \includegraphics[width=\textwidth]{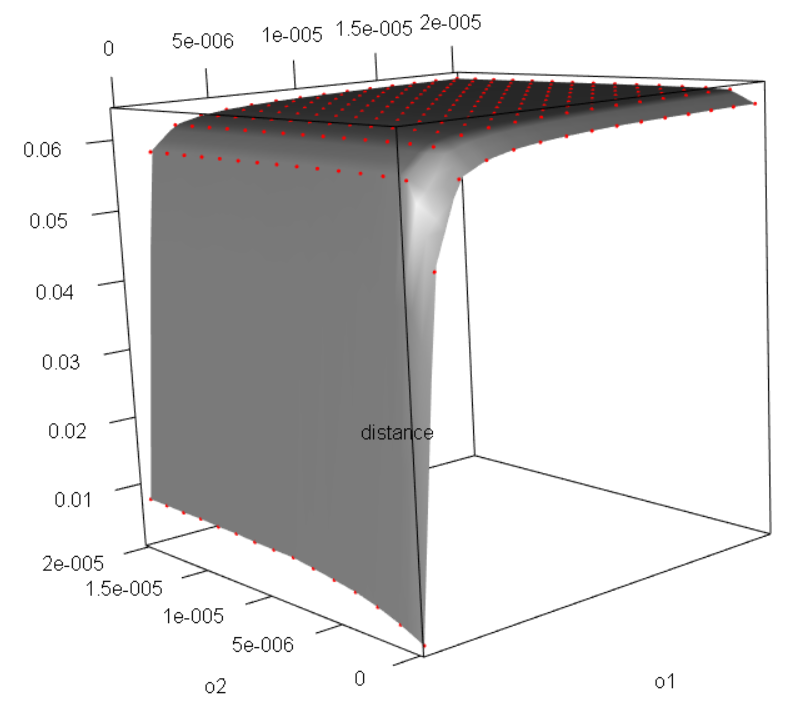}
            \caption{$4$ stocks, FB in and view opposite of reality}
        \end{minipage}
\end{figure}

\begin{figure}[ht]
        \begin{minipage}[b]{0.29\linewidth}
            \centering
            \includegraphics[width=\textwidth]
{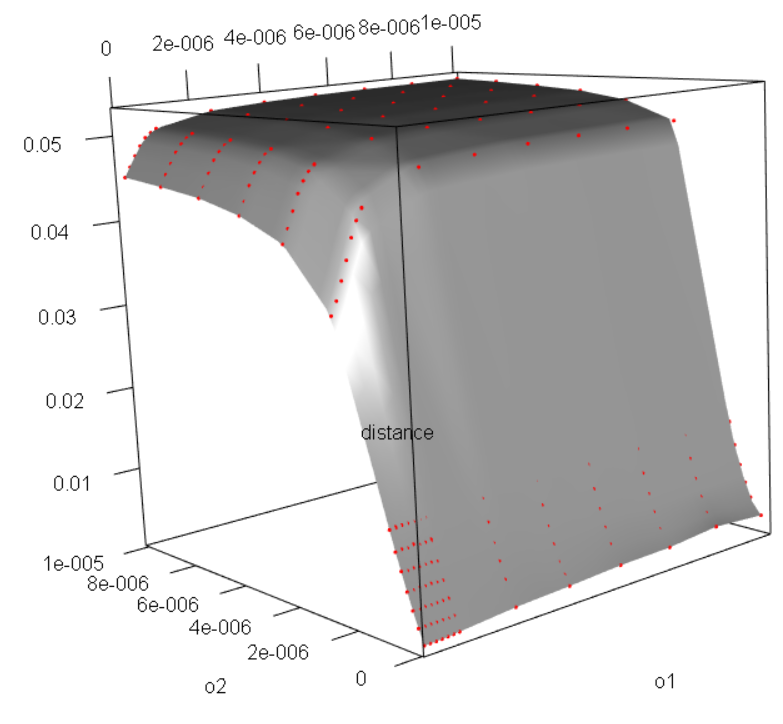}
            \caption{$S\&P 500$, FB in and $q={[0.02,0.05]}^T$}
        \end{minipage}
        \hspace{0.5cm}
        \begin{minipage}[b]{0.29\linewidth}
            \centering
            \includegraphics[width=\textwidth]{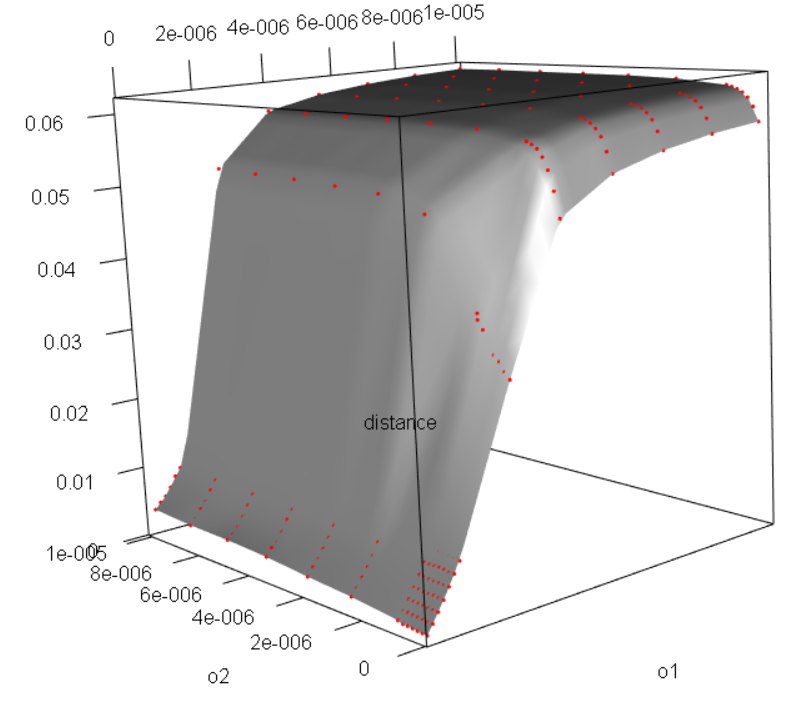}
            \caption{$S\&P 500$, FB in and view exactly like reality}
        \end{minipage}
        \hspace{0.5cm}
        \begin{minipage}[b]{0.29\linewidth}
            \centering
            \includegraphics[width=\textwidth]{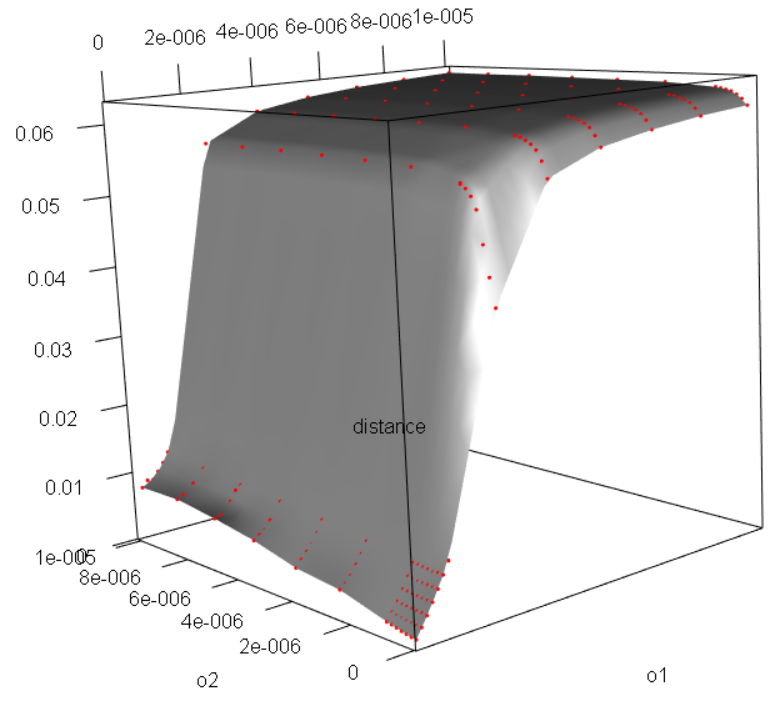}
            \caption{$S\&P 500$ FB in and view opposite of reality}
        \end{minipage}
\end{figure}

Again, just like before, we notice that, as $\omega_i$ get smaller and smaller, when taking into account the whole $S\&P 500$, the curve seems to be under and closer to $0$ than the one when taking into account only $4$ stocks. This might be because the prior on the covariance matrix containing the whole $S\&P 500$ has more information than the one which only has $4$ stocks. Moreover, for the same $q$, the curves have a similar orientation and general shape. Hence, this confirms the belief that albeit a small number of iterations was used for the Gibbs Sampler that takes into account the whole $S\&P 500$, the estimated posterior mean is still accurate. However, as mentioned before, the estimate for $\Sigma_{post}^{-1}$ when it's size is big is not accurate enough to have very reliable profit estimates. 

Nevertheless, for completeness of this analysis, we proceed by leaving all the inputs mentioned before unchanged and keeping $q=\begin{bmatrix}
0.02\\
0.05
\end{bmatrix}$. When taking into account the whole $S\&P 500$, the average profit over the before mentioned range of simulated pairs $(\omega_1,\omega_2)$ is $11,619.97\$$ with a standard deviation of $2,852.246\$$. In the next plot we can observe the profits obtained when considering just the $4$ stocks mentioned before.

\begin{figure}[ht]
	\centering
        \begin{minipage}[b]{0.6\linewidth}
            \centering
            \includegraphics[width=\textwidth]
{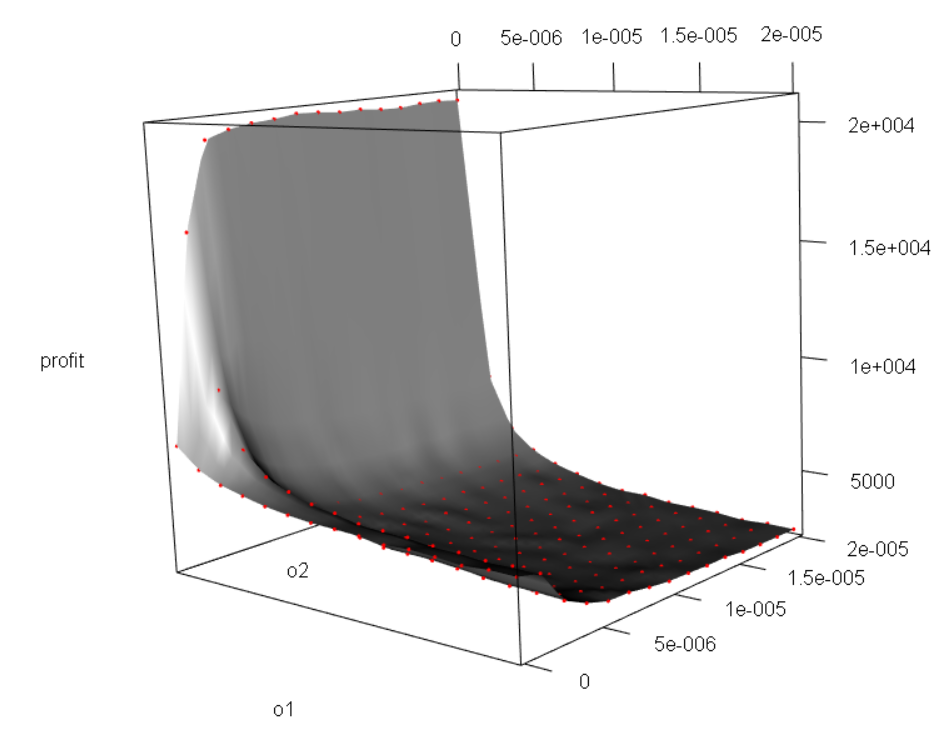}
            \caption{Profit $4$ stocks, FB in and $q={[0.02,0.05]}^T$}
        \end{minipage}
\end{figure}

From the figure, one can see that the first view has a higher influence on the profits than the second view. This is because if we let $\omega_2$ constant the resulting curve increases a lot faster than the curve obtained by keeping $\omega_1$ constant.

\newpage

\section{Prior on $log(\Sigma)$} \label{log(sigma) prior}
\subsection{Introduction}

Just like when introducing the approach with an Inverse Wishart prior, let us see what we would like to improve on it: 
\begin{itemize}
\item	Ideally, the matrix of views $P$ should not be augmented or changed in any way. It should be left just like the user inputted it. 
\item	One should try a different prior than the Inverse Wishart, which has been used numerous times. 
\end{itemize}

Therefore, two of the assumptions will be unchanged: 
\begin{gather*}
r_1,r_2,...,r_m|\mu,\Sigma \sim N_n(\mu,\Sigma) \\
P\mu \sim N_k(q,\Omega) \\
\end{gather*}

A very interesting idea for a different prior on the covariance matrix is presented by Leonard and Hsu (1992)\cite{hsu1992}. As the title of this section is hinting, this prior will actually be on $log(\Sigma)$. In order to better understand Leonard and Hsu's idea, let us look at the distribution: 

\begin{gather*}
f(r_1,...,r_m|\mu,\Sigma)=(2\pi)^{-\frac{mn}{2}}det(\Sigma)^{-\frac{m}{2}}exp \left \{-\frac{1}{2} \sum_{i=1}^{m} (r_i-\mu)^T\Sigma^{-1}(r_i-\mu)  \right \}
\end{gather*}

Let $A=log(\Sigma)$, ${\lambda_{A}}_i$ and ${\lambda_{\Sigma}}_i$ (with $i=\overline{1,n}$) be the eigenvalues of $A$ and $\Sigma$ respectively. Since $A=log(\Sigma)$ we obtain that ${\lambda_{A}}_i=log({\lambda_{\Sigma}}_i)\Rightarrow {\lambda_{\Sigma}}_i=e^{{\lambda_{A}}_i}$. Finally, by remembering that the determinant is the product of the eigenvalues and that the trace of a matrix is the sum of the eigenvalues, we notice that $det(\Sigma)=\prod_{i=1}^{n} {\lambda_{\Sigma}}_i=\prod_{i=1}^{n} e^{{\lambda_{A}}_i}=e^{Tr(A)}$. By using this in the joint distribution of the returns and by noticing that $(r_i-\mu)^T\Sigma^{-1}(r_i-\mu)\in \mathbb{R}$ we obtain: 

\begin{gather*}
f(r_1,...,r_m|\mu,\Sigma)=(2\pi)^{-\frac{mn}{2}} exp \left \{-\frac{1}{2} Tr \left( \sum_{i=1}^{m} (r_i-\mu)^T\Sigma^{-1}(r_i-\mu) \right)   \right \} \\
exp \left\{ -\frac{m}{2} Tr(A) \right\}= (2\pi)^{-\frac{mn}{2}} exp \left \{-\frac{1}{2} Tr \left( \sum_{i=1}^{m} (r_i-\mu)(r_i-\mu)^T\Sigma^{-1} \right)   \right \}\\
exp \left\{ -\frac{m}{2} Tr(A) \right\}=(2\pi)^{-\frac{mn}{2}} exp \left \{-\frac{m}{2} Tr \left(A+ Se^{-A} \right)   \right \}
\end{gather*}
Here, $S=\frac{1}{m}\sum_{i=1}^{m} (r_i-\mu)^T\Sigma^{-1}(r_i-\mu)$.
Before we continue, let us define an operator and make a few notations. 
\begin{definition}
Let $A$ be a $n\times n$ matrix, $A=(a_{ij})_{i,j=\overline{1,n}}$, then we define an operator that stacks in a vector the entries parallel to the main diagonal:
\begin{gather*}
Vec^*(A)=
\begin{bmatrix}
a_{11} & a_{22} & ... &a_{nn}|& a_{12} & a_{23} &...&a_{n-1n}|...|a_{1n}
\end{bmatrix}^T
\end{gather*}
\end{definition}
We notice that if $A$ is $n\times n$, $Vec^*(A)$ is $\frac{n(n+1)}{2}\times 1$. This definition brings us to the following notations:
\begin{notation} \label{notation 1}
\begin{gather*}
\lambda=Vec^*(log(S)), \alpha=Vec^*(log(\Sigma)) \\
\Lambda=log(S), A=log(\Sigma), d=\frac{n(n+1)}{2}
\end{gather*}
\end{notation}

The idea that Leonard and Hsu had was to approximate $f(r_1,...,r_m|\mu,\Sigma)$ by approximating $e^{-A}$. The approximation makes use of the fact that $X(\omega)=e^{-A\omega}$ satisfies a Volterra integral equation\cite{bellman1970}:
\begin{gather*}
X(t)={S}^{-t}-\int_{0}^{t}S^{s-t}(A-\Lambda)X(v)dv, 0<t<\infty, 
\end{gather*}
By letting $t=1$, by iterative substitution of $X(v)$ and by using the spectral decomposition of matrix $S$ we obtain that the approximation is (please see Appendix \ref{approx proof} for the proof): 

\begin{gather}  \label{approx dist}
f^*(r_1,...,r_m|\alpha)=(2\pi e)^{-\frac{mn}{2}}det(S)^{-\frac{m}{2}} exp \left\{ -\frac{1}{2} (\alpha-\lambda)^T Q (\alpha-\lambda) \right\}
\end{gather} 
In order to see how to compute $Q$, we first have to introduce a couple more notations. If we let $e_i,d_i$ to be the $i^{th}$ normalized eigenvector with its corresponding eigenvalue, respectively, then $f_{ij}$ is obtained by looking at the equation $Vec^*(log(\Sigma))^T f_{ij}=e_i^T log(\Sigma) e_j$ and identifying the coefficients of the entries in the $log(\Sigma)$ matrix. With those $f_{ij}$, we can finally compute $Q$: 
\begin{gather*}
\xi_{ij}=\frac{(d_i-d_j)^2}{d_id_j(log(d_i)-log(d_j))^2} \\
Q=\frac{m}{2}\sum_{i=1}^{n} f_{ii}f_{ii}^T+m\sum_{i<j}^{n} \xi_{ij}f_{ij}f_{ij}^T
\end{gather*}
\begin{remark}
The approximate distribution is: $\alpha|r_1,...,r_m \approx\sim N(\lambda,Q^{-1})$
\end{remark}
Now we are ready to move on to the next section and resent the assumptions of the model.

\subsection{The Model}
As mentioned in the previous section, we will have a prior on the $log(\Sigma)$. But how would one construct an intuitive distribution? The simplest distribution that one could work with is the multivariate normal, in which the variance terms on the main diagonal have a mean $\theta_1$ and a variance $\sigma_1^2$ and the covariance terms, which are on the off diagonal, have another mean $\theta_2$ and another variance $\sigma_2^2$. Hence, we arrive at the following model: 
\begin{equation} \label{ass2.1}
r_1,...,r_m|\mu,\Sigma\sim N(\mu,\Sigma)
\end{equation}
\begin{equation}\label{ass2.2}
P\mu\sim N(q,\Omega) 
\end{equation}
\begin{equation}\label{ass2.3}
\alpha|\theta,\Delta=Vec^*(log(\Sigma))|\theta,\Delta\sim N(J\theta,\Delta)
\end{equation} 
Where we have the following notations: 
\begin{notation}
\begin{gather*}
J=\begin{bmatrix}
1&0\\
:&:\\
1&0\\
0&1\\
:&:\\
0&1
\end{bmatrix},
\Delta=\begin{bmatrix}
\sigma_1^2 I_n & \mathbb{O} \\
\mathbb{O} & \sigma_2^2 I_{d-n}
\end{bmatrix}, 
\theta=\begin{bmatrix}
\theta_1\\
\theta_2
\end{bmatrix}
\end{gather*}
\end{notation}

\subsection{Derivation of Posterior Distributions}

If we let $\theta$ to have a uniform prior ($\theta\propto 1$) by integrating it out from the density in equation (\ref{ass2.3}), we obtain: 
\begin{proposition} \label{integral proposition}
\begin{gather*}
f(\alpha|\sigma_1^2,\sigma_2^2)=\int_\theta det(\Delta)^{-\frac{1}{2}}exp\left\lbrace -\frac{1}{2}(\alpha-J\theta)^T \Delta^{-1}(\alpha-J\theta) \right\rbrace d\theta=\\
=2\pi det(\Delta)^{-\frac{1}{2}}det(J^T\Delta^{-1}J)^{-\frac{1}{2}}exp \left\lbrace -\frac{1}{2}\alpha^TG\alpha  \right\rbrace \text{, where} \\
G=\left(I_d-J(J^T\Delta^{-1}J)^{-1}J^T\Delta^{-1}\right)^T\Delta^{-1}\left(I_d-J(J^T\Delta^{-1}J)^{-1}J^T\Delta^{-1}\right)
\end{gather*}
\end{proposition}
For the proof, please see the Appendix \ref{integral proof}.

Now, by using this distribution together with the approximation obtained from the Volterra integral of the distribution of returns denoted by equation (\ref{approx dist}) and with the prior on $P\mu$ represented by equations (\ref{ass2.2}), we can finally obtain the approximate joint distribution: 
\begin{equation} \label{joint dist for log}
\begin{aligned} 
f(\alpha,\mu,\sigma_1^2,\sigma_2^2,r_1,...,r_m)\approx\propto det(\Delta)^{-\frac{1}{2}}det(J^T\Delta^{-1}J)^{-\frac{1}{2}}exp \left\lbrace -\frac{1}{2}\alpha^TG\alpha  \right\rbrace \cdot\\
\cdot det(S)^{-\frac{m}{2}}exp \left\{ -\frac{1}{2} (\alpha-\lambda)^T Q (\alpha-\lambda) \right\}\cdot\\
\cdot det(\Omega)^{-\frac{1}{2}} exp \left\{ -\frac{1}{2} (P\mu-q)^T\Omega^{-1} (P\mu-q)\right\}
\end{aligned}
\end{equation}

We will first proceed with finding the posterior of $\alpha$. Hence, we have to collect all the terms depending on $\alpha$. Since one of those is the approximation obtained from the Volterra integral, the posterior is going to be an approximate distribution: 
\begin{gather*}
f^*(\alpha|r_1,...,r_m,\sigma_1^2,\sigma_2^2,\mu)\approx\propto exp \left\lbrace -\frac{1}{2}\left(\alpha^TG\alpha + (\alpha-\lambda)^T Q (\alpha-\lambda) \right) \right\rbrace
\end{gather*}
We can apply \textbf{Lemma \ref{lemma2} (Completing the square)} with $y=\alpha, a=0,A=G,b=\lambda,B=Q$ and we obtain that:

\begin{gather*}
\boxed{
\alpha|r_1,...,r_m,\sigma_1^2,\sigma_2^2,\mu\approx\sim N(\alpha^*,(Q+G)^{-1}) \text{, where } \alpha^*=(Q+G)^{-1}Q\lambda}
\end{gather*}

Moving to the posterior of $\sigma_1^2,\sigma_2^2$, we have to collect the terms depending on $\Delta$, which also includes $G$. We note that the term obtained from the Volterra integral approximation of the matrix exponential does not show up in this posterior. Hence, this will be an exact distribution: 
\begin{gather*}
f(\sigma_1^2,\sigma_2^2|\alpha,\mu,r_1,...,r_m)\propto det(\Delta)^{-\frac{1}{2}}det(J^T\Delta^{-1}J)^{-\frac{1}{2}}exp \left\lbrace -\frac{1}{2}\alpha^TG\alpha  \right\rbrace
\end{gather*} 

However, one can write the above distribution in scalar form. By applying \textbf{Lemma \ref{determinants identity}} which can be found in \textbf{Appendix \ref{integral proof}}, one finds that the joint posterior distribution of $\sigma_1^2,\sigma_2^2$ is equal to: 
\begin{gather*}
f(\sigma_1^2,\sigma_2^2|\alpha,\mu,r_1,...,r_m)\propto \left( \sigma_1^2 \right)^{-\frac{n-1}{2}}\left( \sigma_2^2 \right)^{-\frac{d-n-1}{2}} exp \left\lbrace -\frac{1}{2}\alpha^TG\alpha  \right\rbrace
\end{gather*}

Furthermore, by applying \textbf{Lemma \ref{exponential identity}} which can also be found in \textbf{Appendix \ref{integral proof}}, we obtain that the scalar version for the equation is: 
\begin{gather*}
f(\sigma_1^2,\sigma_2^2|\alpha,\mu,r_1,...,r_m)\propto \left( \sigma_1^2 \right)^{-\frac{n-1}{2}} exp \left\lbrace - \frac{1}{2\sigma_1^2}\sum_{i=1}^{n} (\alpha_i-\overline{\alpha_v})^2 \right\rbrace \\ 
\left( \sigma_2^2 \right)^{-\frac{d-n-1}{2}}exp \left\lbrace-\frac{1}{2\sigma_2^2}\sum_{i=n+1}^{d} (\alpha_i-\overline{\alpha_c})^2 \right\rbrace
\end{gather*}
Here, $\overline{\alpha_v}$ are the averages of the log of the variance terms and $\overline{\alpha_c}$ are the averages of the log of the covariance terms: 
\begin{gather*}
\overline{\alpha_v}=\frac{\sum_{i=1}^n\alpha_i}{n} \text{ and }  
\overline{\alpha_c}=\frac{\sum_{i=n+1}^d\alpha_i}{d-n}
\end{gather*}

Hence, both posteriors of $\sigma_1^{2}$ and $\sigma_2^{2}$ are following Inverse Gamma distributions and they are independent: 
\begin{gather*}
\boxed{
\sigma_1^{2}|\alpha,\mu,r_1,...,r_m\sim \Gamma^{-1} \left( \frac{n-3}{2},\frac{1}{2}\sum_{i=1}^{n} (\alpha_i-\overline{\alpha_v})^2 \right)}\\
\boxed{
\sigma_2^{2}|\alpha,\mu,r_1,...,r_m\sim \Gamma^{-1} \left( \frac{d-n-3}{2},\frac{1}{2}\sum_{i=n+1}^{d} (\alpha_i-\overline{\alpha_c})^2 \right)}
\end{gather*}

We are finally ready to compute the posterior for $\mu$ also by collecting the terms that depend on it. We notice that the term obtained from the Volterra integral approximation of the matrix exponential does not show up in the posterior. Therefore, like the posteriors of $\sigma_1^2$ and $\sigma_2^2$, this will be an exact distribution. Moreover, we notice that the first two equations in the assumptions of our model (equations (\ref{ass2.1}) and (\ref{ass2.2})) are the same as when we used an Inverse Wishart prior. Therefore, the derivation for the posterior for $\mu$ will be the same, yielding: 

\begin{gather*}
\boxed{
\mu|\alpha,\sigma_1^2,\sigma_2^2,r_1,...,r_m \sim N\left(\mu_{post}, \Sigma_{post} \right)} ,\text{ where } \\
\boxed{
\mu_{post}=(m\Sigma^{-1}+P^T\Omega^{-1}P)^{-1}(m\Sigma^{-1}\overline{r}+P^T\Omega^{-1}q)} \\
\boxed{
\Sigma_{post}=\left( m\Sigma^{-1}+P^T\Omega^{-1}P \right)^{-1}
}
\end{gather*}

\subsection{Implementation}

Now that we have derived our posteriors, we are ready to implement it, using a Gibbs Sampler. The only difference from before is that we will use a Metropolis-Hastings algorithm for sampling $\alpha$, for which we need the exact posterior distribution. This will be proportional to the distribution obtained from collecting all terms with an $\alpha$ from the joint distribution represented by equation (\ref{joint dist for log}): 

\begin{gather*}
exp \left\lbrace -\frac{1}{2}\alpha^TG\alpha  \right\rbrace
\boxed{det(S)^{-\frac{m}{2}}exp \left\{ -\frac{1}{2} (\alpha-\lambda)^T Q (\alpha-\lambda) \right\}}
\end{gather*}

We have seen that it results in the posterior: 
\begin{gather*}
\alpha|r_1,...,r_m,\sigma_1^2,\sigma_2^2,\mu\approx\sim N(\alpha^*,(Q+G)^{-1}) \text{, where } \alpha^*=(Q+G)^{-1}Q\lambda \\
\pi^*(\alpha|r_1,...,r_m,\sigma_1^2,\sigma_2^2,\mu)\approx\propto exp \left \lbrace -\frac{1}{2}(\alpha-\alpha^*)^T(Q+G)(\alpha-\alpha^*) \right \rbrace
\end{gather*}

This is an approximation since the boxed part is an approximation of the pdf of a multivariate normal using the Volterra integral equation. If we replace this with the exact distribution, we would obtain: 
\begin{gather*}
\pi(\alpha|r_1,...,r_m,\sigma_1^2,\sigma_2^2)\propto exp\left\lbrace -\frac{m}{2} Trace\left(A+Se^{-A}\right)-\frac{1}{2} \alpha^TG\alpha  \right\rbrace
\end{gather*}

The Metropolis-Hastings step at $t^{th}$ iteration would be that we would simulate a candidate value from the approximate posterior distribution: $\widetilde{\alpha}\approx\sim N(\alpha^*,(Q+G)^{-1})$ and we would accept it with probability $min(\rho,1)$, where
\begin{gather*}
\rho=\frac{\pi\left(\widetilde{\alpha}|r_1,...,r_m,{\sigma_1^2}^{(t)},{\sigma_2^2}^{(t)},{\mu}^{(t)}\right)}{\pi\left({\alpha}^{(t)}|r_1,...,r_m,{\sigma_1^2}^{(t)},{\sigma_2^2}^{(t)},{\mu}^{(t)}\right)}\cdot \frac{\pi^*\left({\alpha}^{(t)}|r_1,...,r_m,{\sigma_1^2}^{(t)},{\sigma_2^2}^{(t)},{\mu}^{(t)}\right)}{\pi^*\left(\widetilde{\alpha}|r_1,...,r_m,{\sigma_1^2}^{(t)},{\sigma_2^2}^{(t)},{\mu}^{(t)}\right)}
\end{gather*}

It is useful at this point to remember that because of the notation introduced in \ref{notation 1}, we have a connection between $\pi^*$ and $\pi$ since there is one between $A$ and $\alpha$, namely: 
\begin{gather*}
\alpha=Vec^*(A)
\end{gather*}

Using the Metropolis Hastings step that was just discussed, we arrive at the following Gibbs Sampler: 

\begin{algorithm}
\caption{Gibbs Sampler $log(\Sigma)$}
\begin{algorithmic}[1]
\scriptsize 
\STATE $\alpha^{(t+1)}=
\begin{cases}
\widetilde{\alpha}\sim N\left( \left( Q^{(t)}+G^{(t)} \right)^{-1}Q^{(t)}\lambda^{(t)}, \left( Q^{(t)}+G^{(t)} \right)^{-1}\right) \text{w.p. } min(\rho,1) \\
\alpha^{(t)} \text{otherwise}
\end{cases}$
\STATE  Since $\alpha=Vec^*(log(\Sigma))\Rightarrow \begin{cases}
\text{compute } \Sigma^{(t+1)}=exp \left\lbrace {Vec^*}^{-1}\left( \alpha^{(t+1)} \right) \right\rbrace \\
\text{keep } \Sigma^{(t)}
\end{cases}$
\STATE	$\begin{cases}
{\sigma_1^2}^{(t+1)}\sim \Gamma^{-1} \left( \frac{n-3}{2},\frac{1}{2} \sum_{i=1}^n \left( {\alpha_i}^{(t+1)}-{\overline{{\alpha_v}^{(t+1)}}} 
\right)^2 \right) \\
{\sigma_2^2}^{(t+1)}\sim \Gamma^{-1} \left( \frac{d-n-3}{2},\frac{1}{2} \sum_{i=n+1}^d \left( {\alpha_i}^{(t+1)}-{\overline{{\alpha_c}^{(t+1)}}} \right)^2 \right)
\end{cases}\Rightarrow$

$\Rightarrow \Delta^{(t+1)}=
\begin{bmatrix}
{\sigma_1^2}^{(t+1)}I_n & \mathbb{O} \\
\mathbb{O} & {\sigma_2^2}^{(t+1)}I_{d-n}\\
\end{bmatrix}$
\STATE	Let $\Sigma_{\mu}=\left( m{\Sigma^{(t+1)}}^{-1}+P^T\Omega^{-1}P \right)^{-1}\Rightarrow \mu^{(t+1)}\sim N  \left( \Sigma_{\mu}\left(  m{\Sigma^{(t+1)}}^{-1}\overline{r}+P^T\Omega^{-1}q \right), \Sigma_{\mu}\right)$
\STATE Compute $S^{(t+1)}=\frac{\sum_{i=1}^m \left( r_i-\mu^{(t+1)} \right)\left( r_i-\mu^{(t+1)} \right)^T}{m}$, $\lambda^{(t+1)}=Vec^* \left( log \left( S^{(t+1)} \right) \right)$, ${d_j}^{(t+1)}$ and ${e_j}^{(t+1)}$ the eigenvalue and normalized eigenvector of $S^{(t+1)}$ respectively. 
\STATE	Compute $f_{ij}^{(t+1)}$ by identifying the coefficients of the entries of the $log\left(\Sigma \right)$ matrix from the equation $Vec^*\left( log\left( \Sigma^{(t)}\right) \right){f_{ij}}^{(t+1)}={e_i}^{(t+1)}log \left( \Sigma^{(t)} \right){e_j}^{(t+1)} $
\STATE Compute $\xi_{ij}^{(t+1)}=\frac{({d_i}^{(t+1)}-{d_j}^{(t+1)})^2}{{d_i}^{(t+1)}{d_j}^{(t+1)}\left(log\left( {d_i}^{(t+1)} \right)-log({d_j}^{(t+1)})\right)^2}$
\STATE Compute $Q^{(t+1)}=\frac{m}{2}\sum_{i=1}^{n} {f_{ii}}^{(t+1)}{{f_{ii}}^{(t+1)}}^T+m\sum_{i<j}^{n} {\xi_{ij}}^{(t+1)}{f_{ij}}^{(t+1)}{{f_{ij}}^{(t+1)}}^T$
\STATE	Compute 
\begin{gather*}
G^{(t+1)}=\left(I_d-J(J^T{\Delta^{(t+1)}}^{-1}J)^{-1}J^T{\Delta^{(t+1)}}^{-1}\right)^T{\Delta^{(t+1)}}^{-1} \cdot \\
\cdot \left(I_d-J(J^T{\Delta^{(t+1)}}^{-1}J)^{-1}J^T{\Delta^{(t+1)}}^{-1}\right)
\end{gather*}
\end{algorithmic}
\end{algorithm}

\subsection{Results}

Just like we did before, in this section we will depict the sensitivity of the model to changes in confidence levels ($\omega_i$) in terms of both the distance of the posterior to investor's view and the profits obtained if one would use this model to trade.

Before we delve into the actual results for this version of the model, we notice that remarks (\ref{remark 1}) and (\ref{remark 2}) both hold. Basically, this means that as the diagonal entries in $\Omega$ get smaller, the more confident we are in the views because we have the assumption that $P\mu\sim N(q,\Omega)$. Same assumption points out the fact that the smaller $\Omega$ is, the closer $P\mu$ should be to $q$. Hence, a very small $\Omega$ shows the fact that the investor is very confident in this view and, therefore, the posterior should also be close to $q$. Therefore, the smaller our $\Omega$ is, the closer $P\mu_{\text{post}}$ should be to $q$. In the first part of this section we will present some plots similar to the ones presented before. We will take 2 views and do an exhaustive search over possible combinations of pairs of values for the 2 diagonal entries of $\Omega$ (which are depicted as $2$ axis) and compute the same distance as before:$|P\mu_{\text{post}}-q|$ (which is depicted as $1$ axis). 

We chose the same $4$ stocks (AAPL, AMZN, GOOG, MSFT), and we will use the same data set as when we presented the results in section \ref{results P nonsqr Inv Wish}: daily returns from 1/2/2014 to 12/29/2017. We will use the following inputs (again the columns are in order AAPL, AMZN, GOOG, MSFT and the rows represent the views):
\begin{gather*}
q=\begin{bmatrix}
0.02\\
0.05
\end{bmatrix},
P=\begin{bmatrix}
-1 & 1 & 0 & 0\\
0 & 0 & 1 & -1\\
\end{bmatrix}
\end{gather*}

Just like when we had a $P$ non-square and an Inverse Wishart prior,  in this version of the model, one can use smaller confidence levels than when we were just using an Inverse Wishart prior and the augmented matrix $P$. This time one can choose $\omega_i$ (which were defined as the entries in the main diagonal of $\Omega$) of the order $10^{-7}$ without getting any numerical issues. For the results presented here, we let $(\omega_1,\omega_2)$ range between $10^{-6}$ to $10^{-4}$. 

However, one can imagine that this approach is more computationally expensive than just having an Inverse Wishart prior on $\Sigma$. Therefore, the exhaustive search was run in parallel on multiple cores (each core running the Gibbs Sampler for $1$ pair $(\omega_1,\omega_2)$) and the range itself was split into 4 ranges: 
\begin{gather*}
(\omega_1,\omega_2)\in \begin{cases}
\left( 10^{-6}\rightarrow 10^{-5},10^{-6}\rightarrow 10^{-5} \right)\\
\left( 10^{-5}\rightarrow 10^{-4},10^{-5}\rightarrow 10^{-4} \right)\\
\left( 10^{-6}\rightarrow 10^{-5},10^{-5}\rightarrow 10^{-4} \right)\\
\left( 10^{-5}\rightarrow 10^{-4},10^{-6}\rightarrow 10^{-5} \right)
\end{cases}
\end{gather*}
Each one of those ranges was split into an equally spaced grid of $4^2$ points, each one of those being ran on $1$ core.

The burn period was set to $10^{3}$ and the iterations to $10^{4}$. Albeit those seem relatively small, convergence is actually achieved very fast when $\omega_i$ are small.  

\begin{figure}[ht]
	\centering
        \begin{minipage}[b]{0.6\linewidth}
            \centering
            \includegraphics[width=\textwidth]
{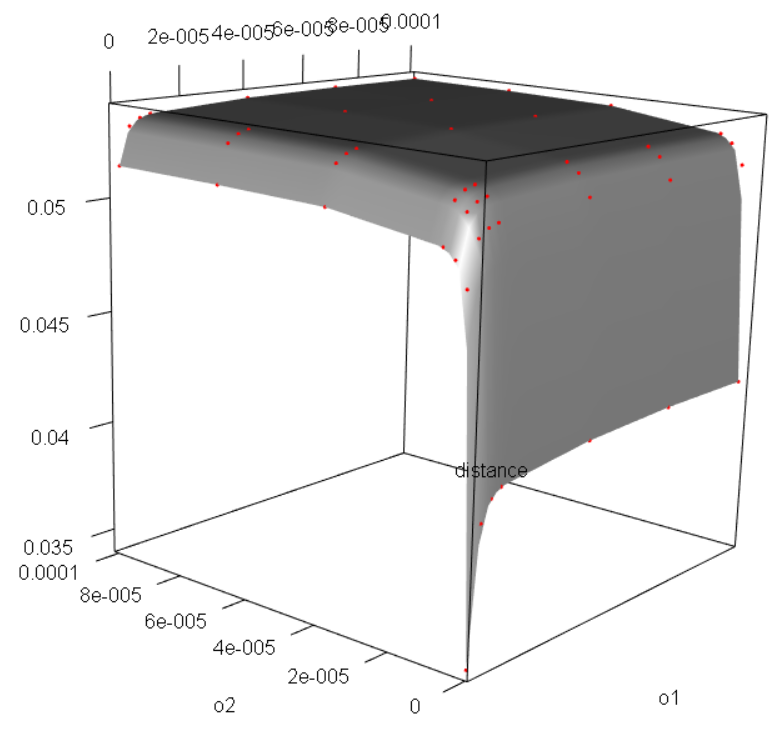}
            \caption{Distances for $log(\Sigma)$ prior}
        \end{minipage}
\end{figure}

We notice that in this version of the model, the distance converges to $0$ very fast as o1 ($\omega_1$ in the model) and o2 ($\omega_2$ in the model) go to $0$. Also, we notice that as o1 and o2 get bigger, it converges very fast to a stabilizing distance. This is consistent with our intuition since if we are very confident in our views, the model should put a lot more importance on them, while if we are not confident at all in our views, the model should just take into consideration the history. Indeed, if we use only the history, the unbiased estimator for $\mu$ is the sample mean of the returns ($\overline{r}$) and therefore the distance becomes $|P\overline{r}-q|=0.05388875$.

We also notice that the second view (corresponding to $o2$) has more influence on the posterior than the first view. This is because the $3D$ curve would leave a $2D$ line on a section parallel to the "o2 vs distance" plane that converges to $0$ as o2 gets very small much faster than a section parallel to the "o1 vs distance" plane would when o1 gets very small.

We will proceed by looking at profits (losses) that we would obtain by using this model trained on the same daily returns between 1/2/2014 and 12/29/2017. We would estimate using Gibbs Sampling the posterior mean ($\mu_{post}$) and the posterior covariance ($\Sigma_{post}$) and we use the CAPM equation (\ref{eq capm}) to obtain the weights to be $w=\frac{1}{2.5}{\Sigma_{post}}^{-1}\mu_{post}$. With those weights we compute the profits that we would obtain over the month January 2018 (just like before, daily returns between 1/2/2018 and 1/30/2018) with an initial investment of $100,000\$$. Here, one could use a different investment horizon also.

The same $P$, $q$, grid for $\omega_i$, burn period, iteration period were used as before. The following is a $3D$ plot of the sensitivity of the profits to changes in confidence: 

\begin{figure}[ht]
	\centering
        \begin{minipage}[b]{0.6\linewidth}
            \centering
            \includegraphics[width=\textwidth]
{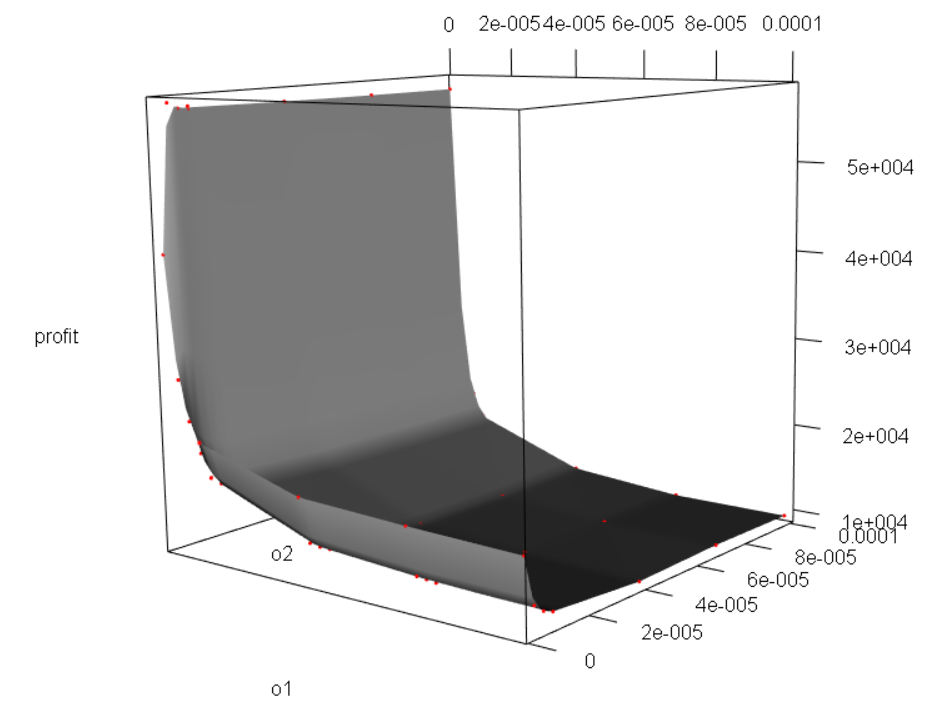}
            \caption{Profits with AMZN in and $q={[0.02,0.05]}^T$}
            \label{fig:figure 6}
        \end{minipage}
\end{figure}

We observe a profit that is approximately between $10,000\$$ and $58,000\$$. In order to interpret this curve, we would have to know what actually happened in the month of January 2018 using the views inputted. More specifically, over the month of January 2018, $Pr_{\text{Jan2018}}=\begin{bmatrix}
0.23996743\\
0.01366718
\end{bmatrix}$. Albeit the inputted $1^{st}$ view is a $10^{th}$ of what happened in reality (AMZN outperformed AAPL by almost $24\%$ in January $2018$), the model puts a higher importance on it than on the $2^{nd}$ view. Indeed, the profits increase drastically as we decrease $\omega_1$ and keep $\omega_2$ constant. Profits do not increase much as we decrease $\omega_2$ and keep $\omega_1$ constant.

Just like we did before, since a $24\%$ gain on AAPL in a month is an extreme scenario, let us consider a different stock instead of AMZN. We will replace AMZN with FB (Facebook) and we will keep all the inputs the same as before, except that we will vary $q$. In the following $3$ figures we will present the results for profits when the investor considers $q=\begin{bmatrix}
0.02\\
0.05
\end{bmatrix}$, $q=\begin{bmatrix}
0.06212815\\
0.01366718
\end{bmatrix}$ which is exactly what happened during the month of January 2018 (the "well informed" investor) and $q=\begin{bmatrix}
-0.06212815\\
-0.01366718
\end{bmatrix}$ which is exactly the opposite of what happened during the month of January 2018 (the "poorly informed" investor):

\begin{figure}[ht]
        \begin{minipage}[b]{0.29\linewidth}
            \centering
            \includegraphics[width=\textwidth]
{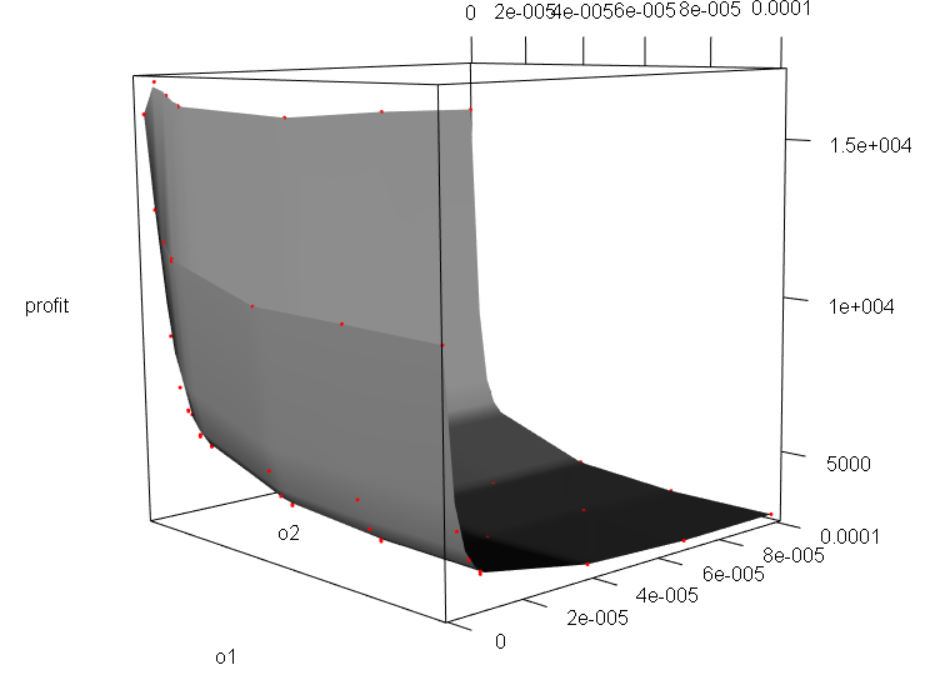}
            \caption{Profits FB instead of AMZN and $q={[0.02,0.05]}^T$}
            \label{fig:figure 7}
        \end{minipage}
        \hspace{0.5cm}
        \begin{minipage}[b]{0.29\linewidth}
            \centering
            \includegraphics[width=\textwidth]{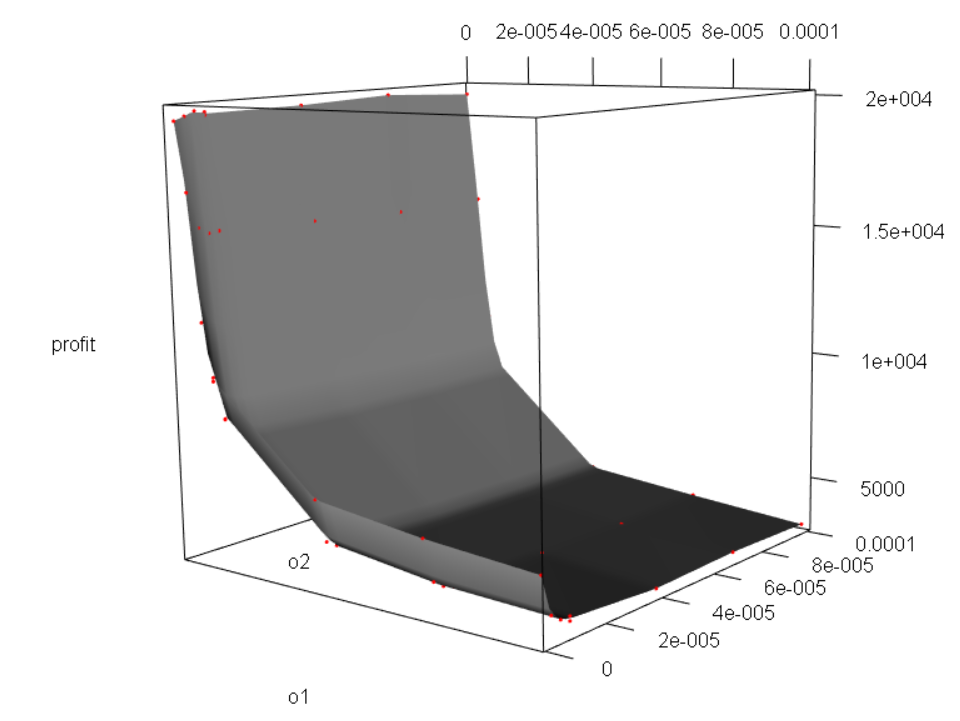}
            \caption{Profits FB instead of AMZN and view exactly like reality}
            \label{fig:figure 8}
        \end{minipage}
        \hspace{0.5cm}
        \begin{minipage}[b]{0.29\linewidth}
            \centering
            \includegraphics[width=\textwidth]{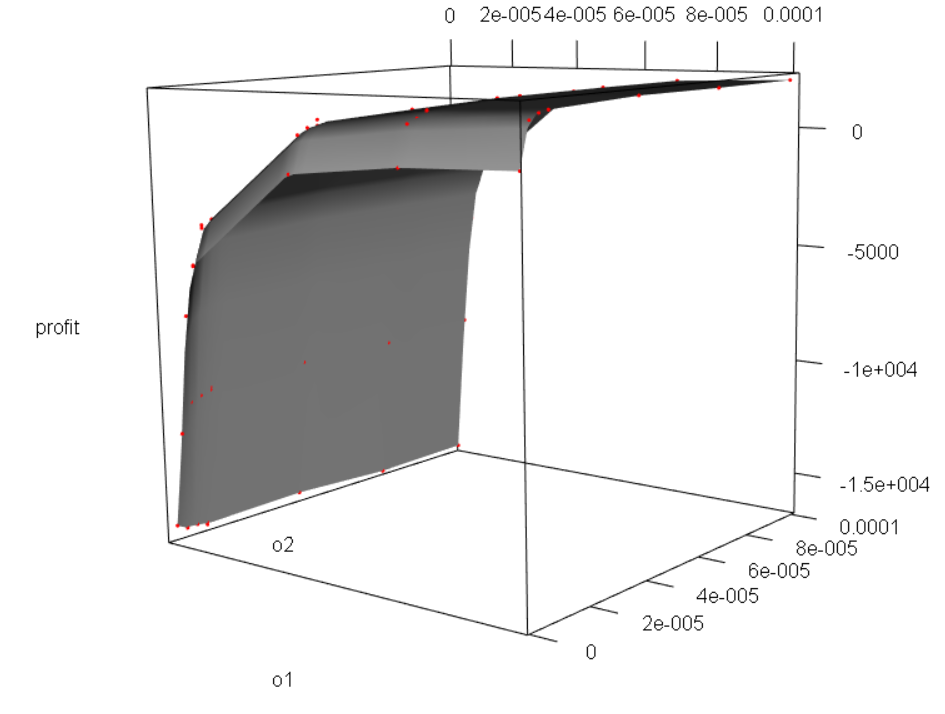}
            \caption{Profits FB instead of AMZN and view opposite of reality}
            \label{fig:figure 9}
        \end{minipage}
\end{figure}

\begin{itemize}
\item	Since $Pr_{\text{Jan2018}}=\begin{bmatrix}
0.06212815\\
0.01366718
\end{bmatrix}$, the view in which $q=\begin{bmatrix}
0.02\\
0.05
\end{bmatrix}$ has returns that are much closer to what happened in reality than when we had AMZN instead of FB (especially the first view is closer). We notice that the second view has a greater influence on the profits than what we have seen in figure \ref{fig:figure 6} and this can be clearly noticed in figure \ref{fig:figure 7} from above.
\item	If the investor has a view exactly like the reality (figure \ref{fig:figure 8}), the first view has more influence on the profits as $\omega_1$ gets smaller and smaller.
\item	Moreover, if we compare figures \ref{fig:figure 8} and \ref{fig:figure 9}, we notice that they seem to be a reflection of each other with respect to a plane parallel to the "o1 vs o2" plane. This would make sense since the only difference between the two is that in figure \ref{fig:figure 8} we have a $q=\begin{bmatrix}
0.06212815\\
0.01366718
\end{bmatrix}$ and in figure \ref{fig:figure 9} we have a $q=-\begin{bmatrix}
0.06212815\\
0.01366718
\end{bmatrix}$. 
\end{itemize}

\subsection{Limitations}

In the previous section, we haven't presented any results for the whole $S\&P500$. This is because we have encountered both memory allocation and running time problems. Both arise from the size of the matrices which makes all matrix computations and sampling from multivariate distributions time consuming. The biggest issue is with the construction of the matrix $Q$. We remind ourselves that we have to compute $f_{ij}$ by looking at the equation $Vec^*(log(\Sigma))^T f_{ij}=e_i^T log(\Sigma) e_j$ and identifying the coefficients of the entries in the $log(\Sigma)$ matrix. With those $f_{ij}$, we can finally compute $Q$: 
\begin{gather*}
\xi_{ij}=\frac{(d_i-d_j)^2}{d_id_j(log(d_i)-log(d_j))^2} \\
Q=\frac{m}{2}\sum_{i=1}^{n} f_{ii}f_{ii}^T+m\sum_{i<j}^{n} \xi_{ij}f_{ij}f_{ij}^T
\end{gather*}

It is easy to compute $\xi_{ij}$ and the elegant way to compute the $f's$ is by coding a $4$ way tensor and applying the function $Vec^*(\cdot)$ to $2$ of its entries (one can see the pattern more easily by taking a small dimensional example). However, this is not the fastest way since one can actually fill out each entry in $Q$ directly. In both situations, the dimensionality problem still exists. When we take into consideration the whole $S\&P 500$, the number of rows and columns are of size $d=\frac{500\cdot 501}{2}$, but since $Q$ is symmetric we would have to store a little more than half of the entries in $Q$ (albeit this approach makes all the formulas in the posterior a lot messier). Even so, the size of such an object is approximately $53$ GB. Even with the biggest server at $UCSB$, for which a node has $1$ TB of RAM memory, we could only run this in parallel on at most $20$ cores. 

The memory allocation problem combined with a running time that is a lot bigger than just the $4$ hours that took to run the simulations presented in section \ref{results P nonsqr Inv Wish} makes this approach computationally not feasible for a large data set.

We have looked at a couple of ideas to remedy the problem: 

\begin{itemize}
\item	Writing the matrix $Q$ to the disk. Unfortunately, one would need a high speed connection (for example SSD) to be able to write it fast enough that it doesn't make the running time even longer. This is of paramount importance since we have to compute $Q$ at each iteration of the Gibbs Sampler.
\item	We have looked at parallelizing the Gibbs Sampler itself (which is a Markov Chain). More precisely, in the general setting of Markov Chains, we have looked at independently starting at $m$ initial points and, from each initial point, starting independent Markov Chains. It has been shown\cite{baum1962} that for one single Markov Chain that satisfies Doob's conditions, the ergodic average converges geometrically: 
\begin{gather*}
P\left(\frac{1}{n} \sum_{k=1}^n f(X_k)>\epsilon\Big|X_0=x_0 \right)\leq A(\epsilon)\rho(\epsilon)^n \text{ ,where}\\
(\exists) d_0,t_0 \text{ s.t. } \rho(\epsilon)=\Phi(d_0,t_0)^{\frac{1}{d_0}}+\eta \text{ with } \eta \text{ s.t. } \rho(\epsilon)<1, \\ \Phi(d_0,t_0)=\sup_{x_0} E \left[ e^{t_0{\sum_{k=1}^{d_0}}f(X_k)}\Big|x_0 \right]
\end{gather*} 
By using this result, one can easily show that for running $m$ Markov Chains in parallel we obtain the following bound: 
\begin{gather*}
P \left( \frac{1}{m} \sum_{i=1}^m \frac{1}{n} \sum_{k=1}^n f(X_{ik})> \epsilon \Big|x_0 \right)\leq e^{-{t_0}^* mn\epsilon}{A^*(\epsilon)}^m{\rho^*(\epsilon)}^{mn}
\end{gather*}
Here, the existence of ${d_0}^*,{t_0}^*$ and the definitions of $A^*(\cdot),\rho^*(\cdot)$ are in the same way as before. The problem is that we cannot compare the right hand sides of the $2$ inequalities from above because the $A(\cdot), A^*(\cdot)$ and $\rho(\cdot), \rho^*(\cdot)$ are different since this is a proof of existence.
\end{itemize}

\subsection{Current Work}

The running time and memory allocation problems encountered when using the whole market would suggest that one has to reduce the dimensionality. Moreover, there is a strong connection between the original Black-Litterman model and CAPM (which can be seen as a factor analysis model in statistics). This gave us the idea of adding a fully Bayesian specified factor model to the Bayesian alternatives presented in this paper. All the posteriors have already been derived for those.

\newpage
\appendix

\section{\\Proof of Original Approach} \label{proof traditional}

As we have seen, the original model is represented by the following $3$ distributions, where the last 2 are priors: 
\begin{gather*}
r \sim N(\mu,\Sigma) \\
\mu \sim N(\pi,\tau \Sigma) \\
q|\mu \sim N(P\mu,\Omega)
\end{gather*}

The last 2 equations are combined. In order to get the joint likelihood , we would have to multiply the probability density functions of the last $2$ distributions: 

\begin{align*}
f(\mu,q) &\propto exp \left \{ -\frac{1}{2} \left[  (\mu-\pi)^T (\tau \Sigma)^{-1}(\mu -\pi)+(q-P\mu)^T\Omega^{-1}(q-P\mu)\right]  \right \}=\\
&=exp\left\{-\frac{1}{2} 
\begin{bmatrix}
\mu-\pi \\
q-P\mu  
\end{bmatrix}^T
\begin{bmatrix}
\tau\Sigma & 0\\
0 & \Omega
\end{bmatrix}^{-1}
\begin{bmatrix}
\mu-\pi \\
q-P\mu  
\end{bmatrix}
\right\}
\end{align*}

Let $V=\begin{bmatrix}
\tau\Sigma & 0\\
0 & \Omega
\end{bmatrix}$ and $\alpha=\begin{bmatrix}
\mu-\pi \\
q-P\mu  
\end{bmatrix}$. 

We will try to find a matrix $A$ such that when we compute $\alpha'=A\alpha$ and $V'=AVA^T$ we will obtain that the joint distribution from above is equal to $exp\left \{-\frac{1}{2}\alpha'^TV'^{-1}\alpha'\right \}$. 

Let 
\begin{align*}
A&=
\begin{bmatrix}
I_n & -\left((\tau\Sigma)^{-1}+P^T\Omega^{-1}P \right)^{-1}P^T\Omega^{-1} \\
0 & I_k
\end{bmatrix}
\begin{bmatrix}
I_n & 0 \\
P & I_k
\end{bmatrix}= \\
&=
\begin{bmatrix}
I_n-\left( (\tau\Sigma)^{-1}+ P^T\Omega^{-1}P \right)^{-1}P^T\Omega^{-1}P & -\left( (\tau\Sigma)^{-1}+ P^T\Omega^{-1}P \right)^{-1}P^T\Omega^{-1} \\
P & I_k 
\end{bmatrix}
\end{align*}

But since $I_n=\left( (\tau\Sigma)^{-1}+ P^T\Omega^{-1}P \right)^{-1}\left( (\tau\Sigma)^{-1}+ P^T\Omega^{-1}P \right)$, we obtain that the first entry in the above matrix can be written as 
\begin{gather*}
\left( (\tau\Sigma)^{-1}+ P^T\Omega^{-1}P \right)^{-1} \left( (\tau\Sigma)^{-1}+ P^T\Omega^{-1}P-P^T\Omega^{-1}P  \right)=\\
\left( (\tau\Sigma)^{-1}+ P^T\Omega^{-1}P \right)^{-1}(\tau\Sigma)^{-1}
\end{gather*}

Hence, we can get the final form of $A$ is: 
\begin{align*}
A=
\begin{bmatrix}
\left( (\tau\Sigma)^{-1}+ P^T\Omega^{-1}P \right)^{-1}(\tau\Sigma)^{-1} & -\left( (\tau\Sigma)^{-1}+ P^T\Omega^{-1}P \right)^{-1}P^T\Omega^{-1} \\
P & I_k
\end{bmatrix}
\end{align*}

\textbf{Note}: $det(A)=1$. The determinant of a $2 \times 2$ block matrix is computed using the same formula as that for a normal $2 \times 2$ matrix, but we have to be careful when we look at the order in which we multiply: 
\begin{align*}
det(A)&=det\{ \left((\tau\Sigma)^{-1}+ P^T\Omega^{-1}P \right)^{-1}(\tau\Sigma)^{-1}+ \\
&+\left( (\tau\Sigma)^{-1}+ P^T\Omega^{-1}P \right)^{-1}P^T\Omega^{-1}P\}= \\
&=det\{\left((\tau\Sigma)^{-1}+ P^T\Omega^{-1}P \right)^{-1}\left( (\tau\Sigma)^{-1}+ P^T\Omega^{-1}P \right)\}=det(I)=1
\end{align*}

Finally, we are ready to find our $\alpha'$ and $V'$: 
\begin{gather*}
\alpha'=A\alpha=A\begin{bmatrix} \mu-\pi \\q-P\mu\end{bmatrix}= \\
\begin{bmatrix}
\left((\tau\Sigma)^{-1}+ P^T\Omega^{-1}P \right)^{-1} (\tau\Sigma)^{-1}(\mu-\pi)-\left((\tau\Sigma)^{-1}+ P^T\Omega^{-1}P \right)^{-1}P^T\Omega^{-1}(q-P\mu)\\
P\mu-P\pi+q-P\mu
\end{bmatrix}
\end{gather*}

Now by factoring the common term in the $1^{st}$ entry of the vector, we obtain: 

\begin{gather*}
\left((\tau\Sigma)^{-1}+ P^T\Omega^{-1}P \right)^{-1}\left( (\tau\Sigma)^{-1}\mu-(\tau\Sigma)^{-1}\pi-P^T\Omega^{-1}q+P^T\Omega^{-1}P\mu\right)=\\
=\left((\tau\Sigma)^{-1}+ P^T\Omega^{-1}P \right)^{-1}\left[ \left((\tau\Sigma)^{-1}+P^T\Omega^{-1}P\right)\mu -(\tau\Sigma)^{-1}\pi-P^T\Omega^{-1}q\right]=\\
=\mu- \left((\tau\Sigma)^{-1}+ P^T\Omega^{-1}P \right)^{-1}\left( (\tau\Sigma)^{-1}\pi+P^T\Omega^{-1}q \right)
\end{gather*}

And if we let $\bar{\mu}=\left((\tau\Sigma)^{-1}+ P^T\Omega^{-1}P \right)^{-1}\left( (\tau\Sigma)^{-1}\pi+P^T\Omega^{-1}q \right)$, we obtain that \begin{equation} \label{eq:1}
\alpha'=A\alpha=\begin{bmatrix} 
\mu-\bar{\mu} \\
q-P\pi
\end{bmatrix} 
\end{equation}

Now we are ready to move to the computation of $V'$:
\begin{align*}
V'&=AVA^T=\begin{bmatrix}
\left((\tau\Sigma)^{-1}+ P^T\Omega^{-1}P \right)^{-1} & -\left((\tau\Sigma)^{-1}+ P^T\Omega^{-1}P \right)^{-1}P^T \\
\tau P\Sigma & \Omega
\end{bmatrix}A^T\\
\end{align*}

And now, in order to compute $A^T$ we can use the fact that for any 2 matrices $(AB)^T=B^TA^T$ and also the fact that $\Sigma$ is symmetric, hence $\Sigma^{-1}$ is also symmetric. Moreover, the same reasoning can be done for $\Omega$. Hence, $(P^T\Omega^{-1}P)^T=P^T\Omega^{-1}P$. Now if we look at the first entry first column in the $2x2$ block matrix, we notice that it is actually equal to $(\tau\Sigma)^{-1} \left((\tau\Sigma)^{-1}+ P^T\Omega^{-1}P \right)^{-1}$. Similarly the second row first column in the block matrix is actually $-\Omega^{-1}P\left((\tau\Sigma)^{-1}+ P^T\Omega^{-1}P \right)^{-1}$. 

Finally, we can compute $V'$. The first entry in the first column will be: 
\begin{gather*}
\left((\tau\Sigma)^{-1}+ P^T\Omega^{-1}P \right)^{-1}(\tau \Sigma)^{-1}\left((\tau\Sigma)^{-1}+ P^T\Omega^{-1}P \right)^{-1}+ \\
+\left((\tau\Sigma)^{-1}+ P^T\Omega^{-1}P \right)^{-1}P^T\Omega^{-1}P\left((\tau\Sigma)^{-1}+ P^T\Omega^{-1}P \right)^{-1}=\\
=\left((\tau\Sigma)^{-1}+ P^T\Omega^{-1}P \right)^{-1}\left((\tau\Sigma)^{-1}+ P^T\Omega^{-1}P \right)\left((\tau\Sigma)^{-1}+ P^T\Omega^{-1}P \right)^{-1}=\\
=\left((\tau\Sigma)^{-1}+ P^T\Omega^{-1}P \right)^{-1}
\end{gather*}

Hence, the whole matrix $V'$ will be: 
\begin{align*}
V'&=\begin{bmatrix}
\left((\tau\Sigma)^{-1}+ P^T\Omega^{-1}P \right)^{-1} & 0 \\
0 & \Omega+\tau P \Sigma P^T
\end{bmatrix} \Rightarrow \\
\Rightarrow {V'}^{-1}&=\begin{bmatrix}
\left((\tau\Sigma)^{-1}+ P^T\Omega^{-1}P \right) & 0 \\
0 & (\Omega+\tau P \Sigma P^T)^{-1}
\end{bmatrix}
\end{align*}

Now that we have both the $V'$ and $\alpha'$, it is a lot easier to notice that the joint will be: 
\begin{align*}
f(q,\mu)&\propto exp \left \{ -\frac{1}{2}\alpha^T V^{-1}\alpha \right \}=exp \left \{ -\frac{1}{2} \alpha'^T A^{-T} A^T V'^{-1} AA^{-1}\alpha' \right \}=\\
&=exp \left \{ -\frac{1}{2} \alpha'^TV'^{-1}\alpha' \right \}
\end{align*}

But this is the kernel of a multivariate normal. Also, from equation (\ref{eq:1}) and from the above form of the matrix $V'$, we obtain that :
\begin{gather*}
\mu \sim N\left(\bar{\mu},\left((\tau\Sigma)^{-1}+ P^T\Omega^{-1}P \right)^{-1}\right), \\
\bar{\mu}=\left((\tau\Sigma)^{-1}+ P^T\Omega^{-1}P \right)^{-1}\left( (\tau\Sigma)^{-1}\pi+P^T\Omega^{-1}q \right)
\end{gather*}

Furthermore, from this solution, we obtained that: 
\begin{gather*}
q\sim N\left( P\pi, \left(\Omega+ P (\tau\Sigma) P^T\right) \right)
\end{gather*}

If we denote by $M^{-1}$ the covariance matrix and we take the above as been the new prior on the expected returns: $r|\mu\sim N(\mu,\Sigma)$, we obtain that the joint is: 

\begin{align*}
f(r,\mu)&\propto exp \left\{ -\frac{1}{2} \left( (\mu-\bar{\mu})^T M(\mu-\bar{\mu}) + (r-\mu)^T\Sigma^{-1}(r-\mu)  \right)  \right\}= \\
&= exp \left\{ -\frac{1}{2} \left( (\mu-\bar{\mu})^T M(\mu-\bar{\mu}) + (\mu-r)^T\Sigma^{-1}(\mu-r)  \right)  \right\}
\end{align*}

Using lemma \ref{lemma2} with $y=\mu,a=\bar{\mu},b=r,A=M,B=\Sigma^{-1}$, we can conclude that: 
\begin{align*}
f(r,\mu) &\propto exp \left\lbrace -\frac{1}{2}(\mu-\mu^*)^T(M+\Sigma^{-1})(\mu-\mu^*)\right\rbrace\cdot \\
&\cdot exp \left\lbrace (r-\bar{\mu})^T \left(M^{-1}+\left({\Sigma^{-1}}\right)^{-1}\right)^{-1} (r-\bar{\mu}) \right\rbrace
\end{align*}

Hence the posterior of the returns fallows also a normal distribution: 
\begin{gather*}
r\sim N(\bar{\mu},M^{-1}+\Sigma)
\end{gather*}

\section{\\Proof of Approximation using Volterra Integrals} \label{approx proof}

As mentioned before, Bellman in his book \textit{Introduction to Matrix Analysis} shows an even more general result than what we need. The matrix exponential $x(t)=e^{(A_0+cB_0)t}$ satisfies the Volterra integral equation: 
\begin{gather*}
X(t)=e^{A_0t}+c\int_{0}^{t}e^{A_0(t-s)}B_0X(s)ds, 0<t<\infty
\end{gather*} 

Now if we let in the above equation $A_0=-\Lambda$, $B_0=\Lambda-A$, $c=1$ and remembering that $\Lambda=log(S)$ we obtain: 
\begin{gather*}
X(t)={S}^{-t}-\int_{0}^{t}S^{s-t}(A-\Lambda)X(v)dv, 0<t<\infty, 
\end{gather*}

Since we want to approximate $e^{-A}$, we let in the above equation $t=1$ and we repeatedly replace $X$: 
\begin{gather*}
e^{-A}=X(1)=S^{-1}-\int_{0}^{1} S^{s-1}(A-\Lambda)S^{-s}ds=\\
=S^{-1} -\int_{0}^{1}S^{s-1}(A-\Lambda) \left( S^{-s}-\int_{0}^{s} S^{u-s}(A-\lambda)X(u)du \right)ds=\\
=S^{-1}-\int_{0}^{1} S^{s-1}(A-\Lambda)S^{-s}ds+\int_{0}^{1}\int_{0}^{s}S^{s-1}(A-\Lambda)S^{u-s}(A-\Lambda)\cdot \\
\cdot \left( S^{-u}-\int_{0}^{u} S^{v-u}(A-\lambda)X(v)dv \right) duds \approx \\
\approx S^{-1}-\int_{0}^{1} S^{s-1}(A-\Lambda)S^{-s}ds+\int_{0}^{1}\int_{0}^{s}S^{s-1}(A-\Lambda)S^{u-s}(A-\Lambda)S^{-u}duds
\end{gather*} 

Where this is an approximation because the triple and higher order integrals were ignored. The conditional pdf of the returns is:
\begin{gather*}
f(r_1,...,r_m|\mu,\Sigma)\propto exp\lbrace -\frac{m}{2}Tr(A+Se^{-A})  \rbrace
\end{gather*}

Hence, from the Volterra approximation, by multiplying by $S$ and taking the trace, we obtain:
\begin{gather*}
Tr(Se^{-A})\approx n-\int_{0}^{1} Tr\left( S^s(A-\Lambda)S^{-s} \right)ds+ \\
+\int_{0}^{1}\int_{0}^{s} Tr\left( S^{s}(A-\Lambda)S^{u-s}(A-\Lambda)S^{-u}duds \right)
\end{gather*}

The first integral is easier to compute: 
\begin{gather*}
\int_{0}^{1} Tr\left( S^s(A-\Lambda)S^{-s} \right)ds=\int_{0}^{1} Tr\left(A-\Lambda \right)ds=Tr\left(A-\Lambda \right)
\end{gather*}

The second integral requires more calculations. Before we delve into them, let us write the spectral decomposition of $S$ as $S=E_0D_0E_0^T$. If we define the matrix $log$ through the Taylor series expansion, and by suing the fact that $E_0$ is orthonormal, we obtain that the spectral decomposition of $log(S)$ is $\Lambda=log(S)=E_0log(D_0)E_0^T$. Also, let us make another notation: $B=E_0^T(A-\Lambda)E_0\Rightarrow E_0BE_0^T=A-\Lambda$:
\begin{gather*}
Tr \left( S^s(A-\Lambda)S^{u-s}(A-\Lambda)S^{-u} \right)=Tr \left( (A-\Lambda)S^{u-s}(A-\Lambda)S^{-(u-s)} \right)=\\
=Tr \left( E_0BD_0^{u-s}B D_0^{-(u-s)}E_0^T \right)=Tr \left( BD_0^{u-s}B D_0^{-(u-s)} \right)
\end{gather*}

In order to compute the integral of this \textit{Trace} term, we will try to put it in scalar form: 
\begin{gather*}
BD_0^{u-s}=
\begin{bmatrix}
b_{11}d_{1}^{u-s} & b_{12}d_{2}^{u-s} & ... & b_{1n}d_{n}^{u-s} \\
: & : & ... & :\\
b_{n1}d_{1}^{u-s} & b_{12}d_{2}^{u-s} & ... & b_{1n}d_{n}^{u-s}
\end{bmatrix}
\end{gather*}

For the matrix $BD_0^{-(u-s)}$ we obtain a similar result, the only difference is that $d_{i}^{u-s}$ is replaced by $\frac{1}{d_{i}^{u-s}}$. Also, from the spectral decomposition, please note that $d_{i}$ are the eigenvalues of $S$.

Since we need the $Tr \left( BD_0^{u-s}B D_0^{-(u-s)} \right)$, we will only compute the diagonal entries of this matrix: 
\begin{gather*}
diag \left(BD_0^{u-s}BD_0^{-(u-s)} \right)= \\
\begin{bmatrix}
b_{11}^2+b_{12}b_{21}\left( \frac{d_{2}}{d_{1}} \right)^{u-s}+b_{13}b_{31}\left( \frac{d_{3}}{d_{1}} \right)^{u-s}+ \cdots + b_{1n}b_{n1}\left( \frac{d_{n}}{d_{1}} \right)^{u-s}\\
b_{21}b_{12}\left(\frac{d_1}{d_2}\right)^{u-s}+b_{22}^2+\cdots + b_{2n}b_{n2}\left(\frac{d_n}{d_2}\right)^{u-s} \\
:\\
b_{n1}b_{1n}\left(\frac{d_n}{d_1}\right)^{u-s}+b_{n2}b_{2n}\left(\frac{d_n}{d_2}\right)^{u-s}+\cdots + b_{nn}^2
\end{bmatrix}
\end{gather*}

But we know that $B$ is symmetric. Therefore, we obtain that:
\begin{gather*}
\int_0^1\int_0^s Tr \left(  BD_0^{u-s}B D_0^{-(u-s)} \right)duds= \sum_{i=1}^{n} \int_0^1\int_0^s  b_{ii}^2duds+ \\
+\sum_{i\neq j}^n \int_0^1\int_0^s b_{ij}^2\left(\frac{d_i}{d_j}\right)^{u-s} duds \text{, where we have that}\\
\sum_{i=1}^{n}\int_0^1\int_0^s  b_{ii}^2duds=\sum_{i=1}^{n}\frac{b_{ii}^2}{2} \text{ and also}\\
\sum_{i\neq j}^n \int_0^1\int_0^s b_{ij}^2\left(\frac{d_i}{d_j}\right)^{u-s} duds= \sum_{i\neq j}^n \int_0^1 b_{ij}^2\left(\frac{d_i}{d_j}\right)^{u-s} \cdot \\
\cdot \frac{1}{\log(d_i)-\log(d_j)}\Big|_0^sds=
\sum_{i\neq j} \frac{b_{ij}^2}{\log(d_i)-\log(d_j)}\int_0^1 1-\left( \frac{d_i}{d_j} \right)^{-s}ds=\\
=\sum_{i\neq j} \frac{b_{ij}^2}{\log(d_i)-\log(d_j)} \left( 1-\left(\frac{d_j}{d_i} \right)^s \frac{1}{\log(d_j)-\log(d_i)}  \right)\Big|_0^1=\\
=\sum_{i\neq j} \frac{b_{ij}^2}{\log(d_i)-\log(d_j)} + \sum_{i \neq j} b_{ij}^2\frac{\frac{d_j}{d_i}-1}{(\log(d_i)-\log(d_j))^2}=\\
=\sum_{i< j} \left( \frac{b_{ij}^2}{\log(d_i)-\log(d_j)} + \frac{b_{ji}^2}{\log(d_j)-\log(d_i)} \right) + \sum_{i<j} b_{ij}^2\frac{\frac{d_j}{d_i}+\frac{d_i}{d_j}-2}{(\log(d_i)-\log(d_j))^2}=\\
=0+\sum_{i<j} b_{ij}^2\frac{\frac{d_j}{d_i}+\frac{d_i}{d_j}-2}{(\log(d_i)-\log(d_j))^2}
\end{gather*}
Finally, by adding the two double integrals, we obtain that 
\begin{gather*}
\int_0^1\int_0^s Tr \left(  BD_0^{u-s}B D_0^{-(u-s)} \right)duds=n-Tr(A)+Tr(\Lambda)+\\
+\frac{1}{2}\sum_{i=0}^n b_{ii}^2+\sum_{i<j} b_{ij}^2\frac{\frac{d_j}{d_i}+\frac{d_i}{d_j}-2}{(\log(d_i)-\log(d_j))^2}
\end{gather*}
With the notation of the $\xi_{ij}$ introduced in the paper, we obtain the Volterra approximation represented by equation (\ref{approx dist}).

\section{\\Proof of Proposition \ref{integral proposition}} \label{integral proof}
The following equality holds:
\begin{gather*}
f(\alpha|\sigma_1^2,\sigma_2^2)=\int_\theta det(\Delta)^{-\frac{1}{2}}exp\left\lbrace -\frac{1}{2}(\alpha-J\theta)^T \Delta^{-1}(\alpha-J\theta) \right\rbrace d\theta=\\
=2\pi det(\Delta)^{-\frac{1}{2}}det(J^T\Delta^{-1}J)^{-\frac{1}{2}}exp \left\lbrace -\frac{1}{2}\alpha^TG\alpha  \right\rbrace \text{, where} \\
G=\left(I_d-J(J^T\Delta^{-1}J)^{-1}J^T\Delta^{-1}\right)^T\Delta^{-1}\left(I_d-J(J^T\Delta^{-1}J)^{-1}J^T\Delta^{-1}\right)
\end{gather*}

\begin{proof}
Before we actually attempt to compute the integral, we would like to put all the quantities in scalar form since this would make our life easier. This brings us to the following two lemmas: 

\begin{lemma} \label{determinants identity}
$det(\Delta)^{-\frac{1}{2}}det(J^T\Delta^{-1}J)^{-\frac{1}{2}}=\frac{1}{\sqrt{n(d-n)}}\left( \sigma_1^2 \right)^{-\frac{n-1}{2}}\left( \sigma_2^2 \right)^{-\frac{d-n-1}{2}}$
\end{lemma} 
\begin{proof}
\begin{gather*}
J^T\Delta^{-1}J=\begin{bmatrix}
\frac{1}{\sigma_1^2}&...&\frac{1}{\sigma_1^2}&0&...&0\\
0&...&0&\frac{1}{\sigma_2^2}&...&\frac{1}{\sigma_2^2}
\end{bmatrix}J=\begin{bmatrix}
\frac{n}{\sigma_1^2}&0\\
0&\frac{d-n}{\sigma_2^2}
\end{bmatrix}
\end{gather*}
Hence, we obtain that $det(J^T\Delta^{-1}J)^{-\frac{1}{2}}=\frac{1}{\sqrt{n(d-n)}}\left( \sigma_1^2 \right)^{\frac{1}{2}}\left( \sigma_2^2 \right)^{\frac{1}{2}}$
Also, clearly since $\Delta$ is diagonal, we obtain that: 
\begin{gather*}
det(\Delta)^{-\frac{1}{2}}=\left( \sigma_1^2 \right)^{-\frac{n}{2}} \left( \sigma_2^2 \right)^{-\frac{d-n}{2}}
\end{gather*}
Multiplying the two determinants, we obtain the desired result.
\end{proof}

Now let us turn our attention to writing in scalar form the term in the exponential: 
\begin{lemma} \label{exponential identity}
$\alpha^TG\alpha=\frac{1}{\sigma_1^2}\sum_{i=1}^{n} (\alpha_i-\overline{\alpha_v})^2+\frac{1}{\sigma_2^2}\sum_{i=n+1}^{d} (\alpha_i-\overline{\alpha_c})^2$, where $\overline{\alpha_v}$ is the average of the $\alpha's$ on the main diagonal (i.e. those that originate from the log of the variance terms of the returns) and $\overline{\alpha_c}$ is the average of all the $\alpha$'s that are on the off diagonal (i.e. those that originate from the log of the covariance terms of the returns).
\end{lemma}
\begin{proof}
First of all, one can notice that the formula for $G$ can be simplified for calculation purposes:
\begin{gather*}
G=\left(I_d-J(J^T\Delta^{-1}J)^{-1}J^T\Delta^{-1}\right)^T\Delta^{-1}\left(I_d-J(J^T\Delta^{-1}J)^{-1}J^T\Delta^{-1}\right)=\\
=\Delta^{-1}-\Delta^{-1}J(J^T\Delta^{-1}J)^{-1}J^T\Delta^{-1}-\Delta^{-1}J(J^T\Delta^{-1}J)^{-1}J^T\Delta^{-1}+\\
+\Delta^{-1}J(J^T\Delta^{-1}J)^{-1}J^T\Delta^{-1}=\Delta^{-1}-\Delta^{-1}J(J^T\Delta^{-1}J)^{-1}J^T\Delta^{-1}
\end{gather*}
We remember that we have computed $J^T\Delta^{-1}J$ in lemma \ref{determinants identity}: 
\begin{gather*}
J^T\Delta^{-1}J=\begin{bmatrix}
\frac{n}{\sigma_1^2}&0\\
0&\frac{d-n}{\sigma_2^2}
\end{bmatrix} \text{ and }
\Delta^{-1}J=\begin{bmatrix}
\frac{1}{\sigma_1^2}&0\\
:&:\\
\frac{1}{\sigma_1^2}&0\\
0&\frac{1}{\sigma_2^2}\\
:&:\\
0&\frac{1}{\sigma_2^2}
\end{bmatrix} \Rightarrow \\
\Rightarrow \Delta^{-1}J\left( J^T\Delta^{-1}J \right)=\begin{bmatrix}
\frac{1}{n}&0\\
:&:\\
\frac{1}{n}&0\\
0&\frac{1}{d-n}\\
:&:\\
0&\frac{1}{d-n}
\end{bmatrix}\Rightarrow \\
\Rightarrow \Delta^{-1}J\left( J^T\Delta^{-1}J \right)J^T\Delta^{-1}=\begin{bmatrix}
\frac{1}{n\sigma_1^2}&...&\frac{1}{n\sigma_1^2}&0&...&0\\
:&...&:&:&...&:\\
\frac{1}{n\sigma_1^2}&...&\frac{1}{n\sigma_1^2}&0&...&0\\
0&...&0&\frac{1}{(d-n)\sigma_2^2}&...&\frac{1}{(d-n)\sigma_2^2}\\
:&...&:&:&...&:\\
0&...&0&\frac{1}{(d-n)\sigma_2^2}&...&\frac{1}{(d-n)\sigma_2^2}
\end{bmatrix}
\end{gather*}
Now we just have to subtract this matrix from $\Delta^{-1}$, which is just diagonal, and we can finally compute the desired quantity: 
\begin{gather*}
\alpha^TG\alpha=\sum_{1\leq i\neq j\leq n}\frac{1}{n\sigma_1^2}\alpha_i\alpha_j+\sum_{i=1}^n\frac{n-1}{n\sigma_1^2}\alpha_i^2+\sum_{n+1\leq i\neq j\leq d} \frac{1}{(d-n)\sigma_2^2}\alpha_i\alpha_j+\\
+\sum_{i=n+1}^d \frac{d-n-1}{(d-n)\sigma_2^2}\alpha_i^2
\end{gather*}
By looking at this equation and the one that we have to prove, we realize that if we would manage to show the following identity, we would also prove the lemma:
\begin{gather*}
\sum_{1\leq i\neq j\leq n}\frac{1}{n\sigma_1^2}\alpha_i\alpha_j+\sum_{i=1}^n\frac{n-1}{n\sigma_1^2}\alpha_i^2=\frac{1}{\sigma_1^2}\sum_{i=1}^{n} (\alpha_i-\overline{\alpha_v})^2
\end{gather*}
Let us start from the right hand side:
\begin{gather*}
\sum_{i=1}^{n} (\alpha_i-\overline{\alpha_v})^2=\sum_{i=1}^n \alpha_i^2-n\overline{\alpha}^2=\sum_{i=1}^n \alpha_i^2-\frac{1}{n}\left( \sum_{i=1}^n \alpha_i \right)^2=\\
=\sum_{i=1}^n \alpha_i^2 -\frac{1}{n}\sum_{i=1}^n\alpha_i^2-\frac{1}{n}\sum_{i\neq j}\alpha_i\alpha_j=\sum_{i=1}^n\frac{n-1}{n}\alpha_i^2-\frac{1}{n}\sum_{i\neq j}\alpha_i\alpha_j
\end{gather*}
\end{proof}

Now we finally have all the necessary identities to write the integral in our proposition in scalar form. We would have to prove that:

\begin{gather*}
\int_{\theta_1} exp \left\lbrace -\frac{1}{2\sigma_1^2} \sum_{i=1}^n (\alpha_i-\theta_1)^2 \right\rbrace d\theta_1 \int_{\theta_2} exp \left\lbrace -\frac{1}{2\sigma_2^2} \sum_{i=n+1}^q (\alpha_i-\theta_2)^2 \right\rbrace d\theta_2=\\
=2\pi \frac{\sigma_1}{\sqrt{n}}exp \left\lbrace -\frac{1}{2\sigma_1^2} \sum_{i=1}^{n} (\alpha_i-\overline{\alpha_v})^2 \right\rbrace \frac{\sigma_2}{\sqrt{d-n}} exp \left\lbrace -\frac{1}{2\sigma_2^2} \sum_{i=n+1}^{q} (\alpha_i-\overline{\alpha_c})^2 \right\rbrace
\end{gather*}

Hence, if we manage to show the following identity, we would manage to prove the proposition also: 
\begin{gather*}
\int_{\theta_1} exp \left\lbrace -\frac{1}{2\sigma_1^2} \sum_{i=1}^n (\alpha_i-\theta_1)^2 \right\rbrace d\theta_1=\sqrt{2\pi} \frac{\sigma_1}{\sqrt{n}}exp \left\lbrace -\frac{1}{2\sigma_1^2} \sum_{i=1}^{n} (\alpha_i-\overline{\alpha_v})^2 \right\rbrace
\end{gather*}

Let us start from the left hand side and subtract and add the average $\overline{\alpha_v}$ in each term of the sum from the exponential:
\begin{gather*}
LHS=\int_{\theta_1} exp \left\lbrace -\frac{1}{2\sigma_1^2} \sum_{i=1}^n (\alpha_i-\theta_1)^2 \right\rbrace d\theta_1=\\
=\int_{\theta_1} exp \left\lbrace -\frac{1}{2\sigma_1^2} \sum_{i=1}^n (\alpha_i-\overline{\alpha_v}+\overline{\alpha_v}-\theta_1)^2 \right\rbrace d\theta_1=\\
=exp \left\lbrace -\frac{1}{2\sigma_1^2} \sum_{i=1}^n (\alpha_i-\overline{\alpha_v})^2 \right\rbrace \cdot \\
\cdot \int_{\theta_1} exp \left\lbrace -\frac{1}{2\sigma_1^2} \left( \sum_{i=1}^n(\overline{\alpha_v}-\theta_1)^2 + 2\sum_{i=1}^n (\alpha_i-\overline{\alpha_v})(\overline{\alpha_v}-\theta_1) \right) \right\rbrace d\theta_1=\\
=exp \left\lbrace -\frac{1}{2\sigma_1^2} \sum_{i=1}^n (\alpha_i-\overline{\alpha_v})^2 \right\rbrace \cdot \\ 
\cdot\int_{\theta_1} exp \left\lbrace -\frac{1}{2\sigma_1^2} \left( n(\overline{\alpha_v}-\theta_1)^2 + 2(\overline{\alpha_v}-\theta_1)\sum_{i=1}^n (\alpha_i-\overline{\alpha_v})\right) \right\rbrace d\theta_1=\\
=exp \left\lbrace -\frac{1}{2\sigma_1^2} \sum_{i=1}^n (\alpha_i-\overline{\alpha_v})^2 \right\rbrace \int_{\theta_1} exp \left\lbrace -\frac{1}{2\left(\frac{\sigma_1}{\sqrt{n}}\right)^2} (\overline{\alpha_v}-\theta_1)^2 \right\rbrace d\theta_1
\end{gather*}

Now we recognize that the term inside the integral is close to the density of a normal distribution. Hence, this gives us the idea of doing the change of variables: 
\begin{gather*}
y_1=\frac{\theta_1-\overline{\alpha_v}}{\frac{\sigma_1}{\sqrt{n}}} \Rightarrow dy_1=\frac{\sqrt{n}}{\sigma_1}d\theta_1 \Rightarrow d\theta_1=\frac{\sigma_1}{\sqrt{n}}dy_1\Rightarrow \\
\Rightarrow LHS=\sqrt{2\pi}exp \left\lbrace -\frac{1}{2\sigma_1^2} \sum_{i=1}^n (\alpha_i-\overline{\alpha_v})^2 \right\rbrace \int_{\theta_1} \frac{1}{\sqrt{2\pi}} exp \left\lbrace -\frac{1}{2}y_1^2 \right\rbrace \frac{\sigma_1}{\sqrt{n}} dy_1=\\
=\sqrt{2\pi} \frac{\sigma_1}{\sqrt{n}}exp \left\lbrace -\frac{1}{2\sigma_1^2} \sum_{i=1}^{n} (\alpha_i-\overline{\alpha_v})^2 \right\rbrace
\end{gather*}
As mentioned before, a similar identity can be showed for the second integral that depends solely on $\theta_2$ and this completes the proof of the proposition.

\end{proof}


\newpage

\begin{thebibliography}{9}
\bibitem{baum1962} 
Baum, L.E., Katz, M. and Read, R.R. (Feb 1962) 
\textit{Convergence rates for the Law of Large Numbers}. 
American Mathematical Society, Vol. 102, No. 2, pp. 187-199

\bibitem{bellman1970} 
Bellman, R. (1970)
\textit{Introduction to Matrix Analysis}. 
McGraw-Hill, New York.

\bibitem{litterman2002} 
Guangliang, H. and Litterman, R. (Feb. 2017) 
\textit{The Intuition Behind Black-Litterman Model Portfolios}. 
SSRN, 22 Oct. 2002, Web.

\bibitem{janecek2004} 
Janecek, K. (2004)
\textit{What is a Realistic Aversion to Risk for Real-world Individual Investors}. 
Semanticsscholar, Web.

\bibitem{hsu1992} 
Leonard, T. and Hsu, John S.J. (Dec. 1992) 
\textit{Bayesian Inference for a Covariance Matrix}. 
The Annals of Statistics, vol. 20, no. 4, pp. 1669-1696 

\bibitem{hsu1999} 
Leonard, T. and Hsu, John S.J. (1999) 
\textit{Bayesian Methods: An Analysis for Statisticians and Interdisciplinary Researchers}. 
Cambridge UK, Cambridge UP. 
 
\end{thebibliography}
\end{document}